\documentclass[10pt,a4paper,reqno]{amsart}
\usepackage{accents}
\usepackage{acronym}
\usepackage{amsfonts}
\usepackage{amsmath}
\usepackage{amssymb}
\usepackage{amsthm}
\usepackage{braket}
\usepackage{cite}
\usepackage{color}
\usepackage{geometry}
\usepackage{graphicx}
\usepackage{hyperref}
\usepackage{mathrsfs}
\usepackage{mathtools}
\usepackage{setspace}
\usepackage{tikz}
\usepackage{txfonts}

\newgeometry{top=1.5cm,bottom=1.5cm,outer=1.5cm,inner=1.5cm}

\newcommand{\bse}{\begin{subequations}}
\newcommand{\ese}{\end{subequations}}

\newtheorem{theorem}{Theorem}[section]
\newtheorem{definition}[theorem]{Definition}

\newtheorem{remark}[theorem]{Remark}

\newtheorem{proposition}[theorem]{Proposition}
\newtheorem{example}[theorem]{Example}
\numberwithin{equation}{section}

\DeclareMathOperator{\tr}{tr}

\def\undertilde#1{\mathord{\vtop{\ialign{##\crcr
$\hfil\displaystyle{#1}\hfil$\crcr\noalign{\kern1.5pt\nointerlineskip}
$\hfil\tilde{}\hfil$\crcr\noalign{\kern-6.5pt}}}}}
\def\underhat#1{\mathord{\vtop{\ialign{##\crcr
$\hfil\displaystyle{#1}\hfil$\crcr\noalign{\kern1.5pt\nointerlineskip}
$\hfil\hat{}\hfil$\crcr\noalign{\kern-6.5pt}}}}}
\def\underbar#1{\mathord{\vtop{\ialign{##\crcr
$\hfil\displaystyle{#1}\hfil$\crcr\noalign{\kern1.5pt\nointerlineskip}
$\hfil\bar{}\hfil$\crcr\noalign{\kern-6.5pt}}}}}
\def\undercheck#1{\mathord{\vtop{\ialign{##\crcr
$\hfil\displaystyle{#1}\hfil$\crcr\noalign{\kern1.5pt\nointerlineskip}
$\hfil\check{}\hfil$\crcr\noalign{\kern-6.5pt}}}}}

\def\wt#1{\tilde{#1}}
\def\wh#1{\hat{#1}}
\def\wb#1{\bar{#1}}
\def\th#1{\wh{\wt{#1}}}
\def\hb#1{\wb{\wh{#1}}}
\def\bt#1{\wt{\wb{#1}}}

\def\ut#1{\undertilde{#1}}
\def\uh#1{\underhat{#1}}
\def\ub#1{\underbar{#1}}

\newcommand{\rd}{\mathrm{d}}
\newcommand{\re}{\mathrm{e}}
\newcommand{\ri}{\mathrm{i}}

\newcommand{\rD}{\mathrm{D}}

\newcommand{\cN}{\mathcal{N}}

\newcommand{\bc}{\boldsymbol{c}}
\newcommand{\tbc}{{}^{t\!}\boldsymbol{c}}
\newcommand{\bu}{\boldsymbol{u}}

\newcommand{\tbv}{{}^{t\!}\boldsymbol{v}}

\newcommand{\bC}{\boldsymbol{C}}
\newcommand{\tbC}{{}^{t\!}\boldsymbol{C}}

\newcommand{\bO}{\boldsymbol{O}}

\newcommand{\bP}{\boldsymbol{P}}
\newcommand{\bQ}{\boldsymbol{Q}}

\newcommand{\bU}{\boldsymbol{U}}

\newcommand{\bLd}{\mathbf{\Lambda}}
\newcommand{\tbLd}{{}^{t\!}\boldsymbol{\Lambda}}
\newcommand{\bOa}{\boldsymbol{\Omega}}
\newcommand{\tbOa}{{}^{t\!}\boldsymbol{\Omega}}

\newcommand{\oa}{\omega}

\newcommand{\Oa}{\Omega}

\acrodef{1D}[1D]{one-dimensional}
\acrodef{2D}[2D]{two-dimensional}
\acrodef{2DTL}[2DTL]{two-dimensional Toda lattice}
\acrodef{3D}[3D]{three-dimensional}
\acrodef{ABS}[ABS]{Adler--Bobenko--Suris}
\acrodef{bDT}[bDT]{binary Darboux transform}
\acrodef{BT}[BT]{B\"acklund transform}
\acrodef{BSQ}[BSQ]{Boussinesq}
\acrodef{CAC}[CAC]{consistency-around-the-cube}
\acrodef{DT}[DT]{Darboux transform}
\acrodef{DL}[DL]{direct linearisation}
\acrodef{DLT}[DLT]{direct linearising transform}
\acrodef{DS}[DS]{Drinfel'd--Sokolov}
\acrodef{DDeltaE}[D$\Delta$E]{differential-difference equation}
\acrodef{FG}[FG]{Fordy--Gibbons}
\acrodef{FX}[FX]{Fordy--Xenitidis}
\acrodef{GD}[GD]{Gel'fand--Dikii}
\acrodef{KP}[KP]{Kadomtsev--Petviashvili}
\acrodef{KK}[KK]{Kaup--Kupershmidt}
\acrodef{KdV}[KdV]{Korteweg--de Vries}
\acrodef{HM}[HM]{Hirota--Miwa}
\acrodef{HS}[HS]{Hirota--Satsuma}
\acrodef{MDC}[MDC]{multi-dimensional consistency}
\acrodef{NQC}[NQC]{Nijhoff--Quispel--Capel}
\acrodef{ODE}[ODE]{ordinary differential equation}
\acrodef{ODeltaE}[O$\Delta$E]{ordinary difference equation}
\acrodef{PDE}[PDE]{partial differential equation}
\acrodef{PDeltaE}[P$\Delta$E]{partial difference equation}
\acrodef{RHP}[RHP]{Riemann--Hilebert problem}
\acrodef{SK}[SK]{Sawada--Kotera}
\acrodef{sG}[sG]{sine--Gordon}
\acrodef{YB}[YB]{Yang--Baxter}

\title[Integrable semi-discretisation of the Drinfel'd--Sokolov hierarchies]{Integrable semi-discretisation of the Drinfel'd--Sokolov hierarchies}
\author{Yue Yin}
\address[YY]{School of Mathematical Sciences \\
East China Normal University \\ 500 Dongchuan Road\\ Shanghai 200241 \\ People's Republic of China}
\author{Wei Fu}
\address[WF]{School of Mathematical Sciences and Shanghai Key Laboratory of Pure Mathematics and Mathematical Practice \\
East China Normal University \\ 500 Dongchuan Road \\ Shanghai 200241 \\ People's Republic of China}

\begin{document}

\begin{abstract}
We propose a novel semi-discrete Kadomtsev--Petviashvili equation with two discrete and one continuous independent variables,
which is integrable in the sense of having the standard and adjoint Lax pairs, from the direct linearisation framework.
By performing reductions on the semi-discrete Kadomtsev--Petviashvili equation, new semi-discrete versions of the Drinfel'd--Sokolov hierarchies associated with Kac--Moody Lie algebras $A_r^{(1)}$, $A_{2r}^{(2)}$, $C_r^{(1)}$ and $D_{r+1}^{(2)}$ are successfully constructed.
A Lax pair involving the fraction of $\mathbb{Z}_\mathcal{N}$ graded matrices is also found for each of the semi-discrete Drinfel'd--Sokolov equations.
Furthermore, the direct linearisation construction guarantees the existence of exact solutions of all the semi-discrete equations discussed in the paper,
providing another insight into their integrability in addition to the analysis of Lax pairs.
\end{abstract}


\maketitle

\section{Introduction}\label{S:Intro}

The theory of discrete integrable systems has been well studied within the past decades, leading to a large number of magnificent achievements in this field, cf. e.g. \cite{HJN16}.
The research on integrable discrete equations not only brings new insights into the modern theory of integrable systems,
but also pushes forward the development of many subjects in pure mathematics such as algebraic geometry, Lie algebras, orthogonal polynomials, special functions and random matrices.

There exist many techniques to search for integrable discretisation of differential equations,
among which a very important one is to construct transformations between solutions.
The main idea of such an approach goes back to the theory of orthogonal polynomials,
where the label of the polynomial within the family can be regarded as the discrete variable.
Then the three-point recurrence relation serves as an \ac{ODeltaE} that family of polynomials satisfies,
while the continuous variable for the related second-order \ac{ODE} plays a role of the parameter of the \ac{ODeltaE}.
Such a technique can also be extended to discretise a \ac{PDE}.
To be more precise, the \ac{BT} and the superposition formula are treated as the associated \ac{DDeltaE} and \ac{PDeltaE}, respectively, see e.g. \cite{WE73,LB80}.

\bse\label{KP}
A typical example to illustrate such an idea is possibly the \ac{KP} equation.
The (potential) \ac{KP} equation reads
\begin{align}\label{CCCKP}
\partial_3u=\frac{1}{4}\partial_1^3u+\frac{3}{2}(\partial_1u)^2+\frac{3}{4}\partial_1^{-1}\partial_2^2u,
\end{align}
in which $u=u(x_1,x_2,x_3)$, $\partial_j$ denotes the partial derivative with respect to the continuous argument $x_j$,
and $\partial_j^{-1}$ denotes the inverse of $\partial_j$, i.e. the integration with respect to $x_j$.
The \ac{BT} of the \ac{KP} equation is given by (see \cite{Che75})
\begin{align}\label{DCCKP}
\partial_2(\wt u-u)=\partial_1^2(\wt u+u)+\partial_1(\wt u-u)^2-2p_1\partial_1(\wt u-u),
\end{align}
which maps a given solution $u$  of \eqref{CCCKP} to a `new' one $\wt u$ associated with a B\"acklund parameter $p_1$.
Since equation \eqref{DCCKP} is of Burgers-type for $\wt u$, the new solution may be obtained using a Cole--Hopf transformation.
To construct more complex solutions of \eqref{CCCKP}, we can introduce another \ac{BT} taking the same form of \eqref{DCCKP} with regard to a solution $\wh u$ and the corresponding parameter $p_2$.
The two \ac{BT}s together give rise to
\begin{align}\label{DDCKP}
\partial_1(\wh u-\wt u)=(p_1-p_2+\wh u-\wt u)(u-\wt u-\wh u+\th u).
\end{align}
This equation was introduced in \cite{Che75}, and also appeared in \cite{NCW85} implicitly.
We also note that it plays a role of a master equation that generates the discrete Calogero--Moser model \cite{NP94}.
Equation \eqref{DDCKP} is the superposition formula (also known as the Bianchi permutability) for the \ac{KP} equation,
namely we can construct a `new' solution $\th u$ which takes an algebro-differential expression of the given solutions $u$, $\wt u$ and $\wh u$.
We can further introduce a third solution $\wb u$ generated from $u$ by the \ac{BT} with parameter $p_3$.
Then from the superposition formula \eqref{DDCKP} and its $(p_2,p_3)$- and $(p_3,p_1)$-analogues, we are able to derive a purely algebraic equation (see \cite{NCW85})
\begin{align}\label{DDDKP}
(p_1-\wt u)(p_2-p_3+\bt u-\th u)+(p_2-\wh u)(p_3-p_1+\th u-\hb u)+(p_3-\wb u)(p_1-p_2-\hb u-\bt u)=0.
\end{align}
Equation \eqref{DDDKP} forms a closed-form relation between six solutions of the \ac{KP} equation \eqref{CCCKP}.
If we now consider the potential $u=u(x_1,x_2,x_3;n_1,n_2,n_3)$ and adopt notations of forward and backward shifts as follows:
\begin{align*}
&\wt u\doteq u(n_1+1,n_2,n_3), \quad \wh u\doteq u(n_1,n_2+1,n_3), \quad \wb u\doteq u(n_1,n_2,n_3+1), \\
&\ut u\doteq u(n_1-1,n_2,n_3), \quad \uh u\doteq u(n_1,n_2-1,n_3), \quad \ub u\doteq u(n_1,n_2,n_3-1),
\end{align*}
equations \eqref{DCCKP} and \eqref{DDCKP} can be considered as \ac{DDeltaE}s, and \eqref{DDDKP} turns out to be a \ac{PDeltaE}.
These equations are often referred to as the semi-discrete and fully discrete \ac{KP} equations, respectively,
and are integrable in their own rights from many aspects \cite{NCW85,WC88a,WC88b}, see also \cite{DJM2} for the bilinear interpretation.
\ese

The essence of such a discretisation is actually the introduction of discrete linear dispersions.
For example, the evolutions of the discrete and continuous \ac{KP} equations listed in \eqref{KP} are fully determined by the product of the plane wave factors
\begin{align}\label{KP:PWF}
\rho(k)=\exp\left\{\sum_{j=1}^\infty k^jx_j\right\}\prod_{i=1}^\infty(p_i+k)^{n_i} \quad \hbox{and} \quad
\sigma(k')=\exp\left\{-\sum_{j=1}^\infty (-k')^jx_j\right\}\prod_{i=1}^\infty(p_i-k')^{-n_i},
\end{align}
where $k$ and $k'$ are the two separate spectral parameters, cf. \cite{NCW85,DJM2}.
Here the factors associated with $n_j$ should be understood as discretisations of the exponents of $x_j$.
This is because we can recover the continuous linear dispersion from the discrete one through the Miwa transform (cf. \cite{Miw82})
\begin{align*}
\rho(k)\sigma(k')=\prod_{i=1}^\infty\left(\frac{p_i+k}{p_i-k'}\right)^{n_i}
=\exp\left\{\sum_{j=1}^\infty\left[k^j-(-k')^j\right]\frac{(-1)^{j-1}}{j}\sum_{i=1}^\infty\frac{n_i}{p_i^j}\right\}
\doteq\exp\left\{\sum_{j=1}^\infty\left[k^j-(-k')^j\right]x_j\right\}.
\end{align*}
In other words, if we progressively replace the continuous independent variables $x_j$ by the discrete ones $n_i$ following \eqref{KP:PWF},
the continuous \ac{KP} is discretised, and as a consequence equations \eqref{DCCKP}, \eqref{DDCKP} and \eqref{DDDKP} arise successively.

Discrete equations arising from \ac{BT} and nonlinear superposition have considerable significance in view of their underlying rich algebraic and geometric structures,
leading to new notions such as multi-dimensional consistency \cite{DS97,NW01,BS02}, Lagrangian multiforms \cite{LN09} and discrete Painlev\'e equations \cite{NP91,NRGO01}.
Therefore, it turns to be an interesting problem to search for semi- and fully discrete equations that play roles of \ac{BT} and Bianchi permutability of integrable \ac{PDE}s as many as possible.
By following such an idea, a great progress has been made and a huge class of integrable discrete equations were discovered.
Typical examples include the discrete equations in the \ac{KdV} and \ac{BSQ} families (see e.g. \cite{NQC83,DJM3,NPCQ92,ZZN12}),
or more generally the discrete \ac{GD} hierarchy (see e.g. \cite{NPCQ92,Dol13,FX17a,Fu20}) that contains higher-rank lattice equations,
in addition to the discrete \ac{KP} equation \cite{NCW85,DJM2,Hir81}.
All these equations are associated with the $A$-type Lie algebras, i.e. $A_\infty$ for the \ac{KP} equation and $A_r^{(1)}$ for the discrete \ac{GD} hierarchy.

Then a natural question would be whether there exist integrable discrete equations associated with Lie algebras of other types;
or more deeply, whether there exists a classification of discrete integrable systems based on Lie algebras.
This is motivated by the remarkable observation that integrable \ac{PDE}s can be classified in terms of Lie algebras, see e.g. \cite{DS85,JM83,Wil81,FG83,MOP83,UT84},
which is now often referred to as the Drinfel'd--Sokolov classification.
The study on such a topic induced the discovery of the discrete BKP \cite{Miw82}, CKP \cite{Kas96,Sch03} and DKP \cite{Shi00} equations,
which are named after the infinite-dimensional algebras $B_\infty$ , $C_\infty$ and $D_\infty$, respectively.
However, the picture for the classification of \ac{2D} discrete integrable systems is far from complete,
mainly due to the complexity of \ac{BT}s and permutability for \ac{PDE}s related to Kac--Moody algebras which are not of $A_r^{(1)}$-type.
One successful example is possibly the class of equations of $A_2^{(2)}$-type (also known as the $BC_1$-type in the literature).
In this class the \ac{BT} and the permutability are fully understood \cite{SK77,MV00,HL91,Sch96} for the \ac{SK}, \ac{KK} and Tzitzeica equations,
but were not interpreted as discrete equations (which was illustrated by the \ac{SK} equation) until very recently, see \cite{MLX18}.
It is also worth mentioning that there exist different integrable discretisations of the Drinfel'd--Sokolov hierarchies, see e.g. \cite{TH96,AP11,GHY12,HK19,HK20,HZW00,FX17c}.

We would like to study integrable discretisation of the Drinfel'd--Sokolov hierarchies from a unified perspective through introducing discrete plane wave factors.
The method we adopt is the so-called \ac{DL}.
The \ac{DL} method was invented to solve initial value problems of integrable \ac{PDE}s by Fokas, Ablowitz and Santini \cite{FA81,FA83,SAF84}.
The main idea of this approach is to transform a nonlinear integrable \ac{PDE} into a very general linear integral equation.
Hence, the problem turns out to be solving the integral equation, and subsequently, a very large class of exact solutions for the nonlinear \ac{PDE} are constructed.
The \ac{DL} is not only an effective method of solving initial value problems of integrable \ac{PDE}s,
but more importantly a systematic tool to study the underlying structures of discrete and continuous integrable systems,
including searching for integrable discretisations of nonlinear \ac{PDE}s (see e.g. the review papers \cite{NC95,NCW85,NPCQ92})
and constructing integrability characteristics of discrete equations such as Lax pairs \cite{NPCQ92} and master symmetries \cite{NRGO01,Fu21a}, etc.
The key point of realising these is the introduction of the infinite matrix language in \cite{NQLC83}.
This transfers the \ac{DL} to be a purely algebraic method, playing a crucial role in constructing integrable equations, especially discrete ones.
In contrast to the bilinear approach, the \ac{DL} provides a path towards algebraic construction of nonlinear equations directly,
which brings the advantage that we are not necessarily restricted to the bilinear identity.
This allows us to find more complex discrete equations, see examples such as the discrete \ac{BSQ} and CKP equations \cite{ZZN12,Fu17a}.

In the recent papers \cite{Fu18b,Fu21b}, the connection between the linear integral equations in the \ac{DL} framework
and the continuous Drinfel'd--Sokolov and \ac{2D} Toda hierarchies associated with the infinite-dimensional algebras $A_\infty$, $B_\infty$ and $C_\infty$
as well as the Kac--Moody Lie algebras $A_r^{(1)}$, $A_{2r}^{(2)}$, $C_r^{(1)}$ and $D_{r+1}^{(2)}$ was established.
This gives us a strong hint about how to construct the discrete Drinfel'd--Sokolov hierarchies of these types from the \ac{DL} framework.
Notice that the discrete $A_r^{(1)}$-type equations have been worked out from the \ac{DL} by selecting the plane wave factors \eqref{KP:PWF}, see \cite{NPCQ92,Fu18a,Fu20}.
To search for the rest classes of discrete integrable systems, we need to introduce discrete odd-flow variables,
because integrable \ac{PDE}s of $A_{2r}^{(2)}$-, $C_r^{(1)}$- and $D_{r+1}^{(2)}$-types are fully described by continuous odd-flow variables $x_{2j+1}$, see e.g. \cite{JM83}.
For this reason, we consider the plane wave factors
\begin{align}\label{A:PWF}
\rho_n(k)=\exp\left\{\sum_{j=0}^\infty k^{2j+1}x_{2j+1}\right\}\prod_{i=1}^\infty\left(\frac{p_i+k}{p_i-k}\right)^{n_i}k^n \quad \hbox{and} \quad
\sigma_n(k')=\exp\left\{\sum_{j=0}^\infty k'^{2j+1}x_{2j+1}\right\}\prod_{i=1}^\infty\left(\frac{p_i+k'}{p_i-k'}\right)^{n_i}(-k')^{-n},
\end{align}
motivated by the observation in \cite{Fu17a}.
In fact, the Miwa transform
\begin{align*}
\prod_{i=1}^\infty\left(\frac{p_i+k}{p_i-k}\frac{p_i+k'}{p_i-k'}\right)^{n_i}
=\exp\left\{\sum_{j=0}^\infty\left(k^{2j+1}+k'^{2j+1}\right)\frac{2}{2j+1}\sum_{i=1}^\infty\frac{n_i}{p_i^{2j+1}}\right\}
\doteq\exp\left\{\sum_{j=0}^\infty\left(k^{2j+1}+k'^{2j+1}\right)x_{2j+1}\right\}
\end{align*}
indicates that the factors for $n_i$'s are indeed suitable discretisations of the exponents regarding the continuous temporal arguments $x_{2j+1}$.
Our construction is based on a set of linear integral equations given as follows:
\bse\label{Linear}
\begin{align}
&\bu_n(k)+\iint_D\rd\zeta(l,l')\rho_n(k)\Oa(k,l')\sigma_n(l')\bu(l)=\rho_n(k)\bc(k), \label{Integral} \\
&\tbv_n(k')+\iint_D\rd\zeta(l,l')\rho_n(l)\Oa(l,k')\sigma_n(k')\tbv(l')=\tbc(k')\sigma_n(k'), \label{AdIntegral}
\end{align}
\ese
where $\rho_n(k)$ and $\sigma_n(k')$ are plane wave factors defined as \eqref{A:PWF};
$\Oa(k,k')$ is the kernel of the integral equations given by
\begin{align}\label{A:Kernel}
\Oa(k,k')=\frac{1}{k+k'};
\end{align}
$\rd\zeta(k,k')$ and $D$ are arbitrary integration domain and measure without any restriction;
the wave functions $\bu_n(k)$ and $\tbv_n(k')$ are infinite column and row vectors having components depending on $x_i$, $n_i$ and $n$, as well as the respective spectral parameters $k$ and $k'$;
$\bc(k)$ and $\tbc(k')$ are infinite-dimensional column and row vectors, defined as
\begin{align}\label{c}
\bc(k)\doteq{}^{t\!}(\cdots,k^{-1},1,k,\cdots) \quad \hbox{and} \quad
\tbc(k')\doteq(\cdots,k'^{-1},1,k',\cdots),
\end{align}
respectively; in other words, $\bc(k)$ and $\tbc(k')$ have their corresponding $i$th-components given by $k^i$ and $k'^i$.
Here the notation ${}^{t\!}(\cdot)$ denotes the transpose of an infinite matrix or vector.
As the first paper in our series work to solve the problem of searching for integrable discretisations of the Drinfel'd--Sokolov hierarchies from the \ac{DL} method,
we investigate semi-discrete versions of the Drinfel'd--Sokolov equations associated with Kac--Moody algebras $A_{2r}^{(2)}$-, $C_r^{(1)}$- and $D_{r+1}^{(2)}$-types
by focusing on the flow-variables $x_1$, $n_1$ and $n$ in \eqref{A:PWF}.

The following results are achieved.
A novel semi-discrete equation with two discrete and one continuous independent variables of \ac{KP}-type (see equation \eqref{A:u}) is proposed within the \ac{DL} scheme,
which plays a role of a further discretisation of both the third-order differential-difference \ac{KP} equation and the \ac{2D} Toda equation.
Such an equation is proven integrable in the sense of having the standard and adjoint Lax pairs,
in which a new-type discrete linear problem occurs compared with the spectral problems for the discrete and continuous \ac{KP} equations in \eqref{KP}.
New semi-discrete versions of the Drinfel'd--Sokolov hierarchies associated with $A_r^{(1)}$ (which is a byproduct), $A_{2r}^{(2)}$, $C_r^{(1)}$ and $D_{r+1}^{(2)}$
together with their relevant Lax pairs are successfully constructed from the semi-discrete \ac{KP} equation by reductions in the \ac{DL}.
An interesting observation is that the fractional $\mathbb{Z}_\cN$-graded Lax matrices appear in the reduced cases,
which constitute new discretisations of the \ac{FG} Lax representations in the continuous theory.
Our construction, namely the \ac{DL} approach, also guarantees the existence of exact solutions.
This provides a different insight into the essential integrability of all the proposed \ac{3D} and \ac{2D} semi-discrete integrable systems,
in addition to the analysis of Lax pairs.

We organise the paper as follows.
In section \ref{S:DL}, we give a brief introduction to the infinite matrix language of the \ac{DL} approach.
Section \ref{S:SmDisKP} is contributed to the semi-discrete \ac{KP} equation and its Lax representation.
Reductions to the semi-discrete Drinfel'd--Sokolov equations associated with the Kac--Moody Lie algebras $A_r^{(1)}$, $A_{2r}^{(2)}$, $C_r^{(1)}$ and $D_{r+1}^{(2)}$ are presented in section \ref{S:DS}, illustrated by the simplest nontrivial examples.

\section{Preliminaries}\label{S:DL}

This section is a brief introduction to the \ac{DL} framework.
We first define infinite matrices $\bC_n$ and $\bOa$ by
\begin{align}\label{C}
\bC_n\doteq\iint_D\rd\zeta(k,k')\rho_n(k)\bc(k)\tbc(k')\sigma_n(k')
\end{align}
and
\begin{align}\label{Omega}
\tbc(k')\bOa\bc(k)\doteq\Oa(k,k'),
\end{align}
respectively, in which $\rho_n(k)$, $\sigma_n(k')$ and $\Oa(k,k')$ are given by \eqref{A:PWF} and \eqref{A:Kernel}, respectively.
From the definitions, we observe that the infinite matrices $\bC_n$ and $\bOa$ are in a sense the infinite matrix representations of the effective plane wave factor $\rho_n(k)\sigma_n(k')$ and the Cauchy kernel $\Oa(k,k')$, respectively.
To construct integrable systems in the \ac{DL}, we need an infinite potential matrix.
\begin{definition}
The infinite potential matrix in the \ac{DL} is defined as
\begin{align}\label{Potential}
\bU_n\doteq\iint_D\rd\zeta(k,k')\bu_n(k)\tbc(k')\sigma_n(k'),
\end{align}
where $\bu_n(k)$, $\rd\zeta(k,k')$ and $D$ match the ones in the linear integral equation \eqref{u}, $\tbc(k')$ and $\sigma_n(k')$ are given by \eqref{c} and \eqref{A:PWF}, respectively.
\end{definition}

\begin{proposition}
The wave functions $\bu_n(k)$ and $\tbv_n(k')$ satisfy the following equations:
\bse\label{uv}
\begin{align}
&\bu_n(k)=(1-\bU_n\bOa)\bc(k)\rho_n(k), \label{u} \\
&\tbv_n(k')=\sigma_n(k')\tbc(k')(1-\bOa\bU_n). \label{v}
\end{align}
\ese
\end{proposition}
\begin{proof}
Equations \eqref{uv} are the respective reformulations of the integral equations \eqref{Integral} and \eqref{AdIntegral} in virtue of \eqref{Omega} and \eqref{Potential}.
\end{proof}
\begin{proposition}
The infinite matrix $\bU_n$ satisfies
\begin{align}\label{U}
\bU_n=(1-\bU_n\bOa)\bC_n, \quad \hbox{or alternatively} \quad \bU_n=\bC_n(1+\bOa\bC_n)^{-1}.
\end{align}
\end{proposition}
\begin{proof}
Equation \eqref{U} is obtained by performing the integration $\iint_D\cdot\,\rd\zeta(k,k')$ on \eqref{u}, with the help of \eqref{Potential}.
\end{proof}
\begin{remark}
The infinite matrix $\bU_n$ given by \eqref{U} can equivalently be defined as
\begin{align}\label{AdPotential}
\bU_n\doteq\iint_D\rd\zeta(k,k')\rho_n(k)\bc(k)\tbv_n(k').
\end{align}
This is because such a definition together with \eqref{v} gives rise to $\bU_n=\bC_n(1-\bOa\bU_n)$, which is equivalent to \eqref{U} that follows from the original definition \eqref{Potential}.
\end{remark}
\begin{definition}
The tau function in the \ac{DL} is defined as
\begin{align}\label{tau}
\tau_n\doteq\det(1+\bOa\bC_n), \quad \hbox{or equivalently} \quad \tau_n\doteq\det(1+\bC_n\bOa),
\end{align}
where $\bC_n$ and $\bOa$ are defined by \eqref{C} and \eqref{Omega}, respectively.
\end{definition}
\begin{remark}
The determinant of the infinite matrix in \eqref{tau} should be thought of as the expansion
\begin{align*}
\det(1+\bOa\bC_n)=1+\sum_{i}(\bOa\bC_n)^{(i,i)}+\sum_{i<j}
\left|
\begin{matrix}
(\bOa\bC_n)^{(i,i)} & (\bOa\bC_n)^{(i,j)} \\
(\bOa\bC_n)^{(j,i)} & (\bOa\bC_n)^{(j,j)}
\end{matrix}
\right|+\cdots,
\end{align*}
in which the action $(\cdot)^{(i,j)}$ stands for taking the $(i,j)$-entry of the corresponding infinite matrix.
Notice that all the terms in such an expansion are actually $\tr[(\bOa\bC_n)^i]$ for $i\in\mathbb{N}$.
In order to make sense of the expansion, we also have to impose the condition on integration measure and domain in terms of the spectral variables in the infinite matrix $\bC_n$ defined by \eqref{C} such that $\tr[(\bOa\bC_n)^i]$ truncates, see \cite{NRGO01}.
\end{remark}
The infinite matrix $\bU_n$, the infinite vectors $\bu_n(k)$ and $\tbv_n(k')$, as well as the tau function are the key objects in our approach,
and later we will see that they are the ingredients to construct nonlinear integrable equations, associated linear problems,
and the homogeneous equations of the tau function (which in many (but not all) cases are bilinear), respectively.

In concrete computation in our scheme, we also need infinite matrices $\bLd$, $\tbLd$ and $\bO$.
The respective $(i,j)$-entries of these matrices are defined as
\begin{align}\label{InfMat}
\bO^{(i,j)}=\delta_{i,0}\delta_{0,j}, \quad \bLd^{(i,j)}=\delta_{i+1,j} \quad \hbox{and} \quad \tbLd^{(i,j)}=\delta_{i,j+1},
\end{align}
for all $i,j\in\mathbb{Z}$, in which $\delta_{\cdot,\cdot}$ denotes the standard Kroneker delta function, namely
\begin{align*}
\delta_{i,j}=
\left\{
\begin{array}{ll}
1, & i=j, \\
0, & i\neq j,
\end{array}
\right. \quad
\forall i,j\in\mathbb{Z}.
\end{align*}
Through direct calculation we are able to prove that the infinite matrices $\bLd$, $\tbLd$ and $\bO$ possess properties as follows:
\begin{align*}
&\bLd^i\bc(k)=k^i\bc(k), \quad \tbc(k')\tbLd^j=k'^j\tbc(k'), \quad \tbc(k')\tbLd^j\bO\bLd^i\bc(k)=k^ik'^j, \\
&(\bLd^{i'}\bU_n)^{(i,j)}=\bU_n^{(i+i',j)}, \quad (\bU_n\tbLd^{j'})^{(i,j)}=\bU_n^{(i,j+j')}, \quad (\bU_n\tbLd^{j'}\bO\bLd^{i'}\bU_n)^{(i,j)}=\bU_n^{(i,j')}\bU_n^{(i',j)}, \\
&[\bLd^{i'}\bu_n(k)]^{(i)}=[\bu_n(k)]^{(i+i')}, \quad [\tbv_n(k')\tbLd^{j'}]^{(j)}=[\tbv_n(k')]^{(j+j')},
\end{align*}
where $(\cdot)^{(i)}$ denotes taking the $i$th-component of an infinite vector.
These properties transfer the operations of $\bLd$, $\tbLd$ and $\bO$ to the shifts of the components or entries of infinite vectors and matrices,
which will later play a crucial role in establishing connections between discrete/continuous dynamics and purely algebraic relations of infinite matrices in the construction of integrable models.

\section{A semi-discrete Kadomtsev--Petviashvili equation}\label{S:SmDisKP}

\subsection{Algebraic construction of the closed-form semi-discrete equation}

Our aim in this subsection is to construct closed-form scalar \ac{3D} semi-discrete equations of $x_1$, $n_1$ and $n$ expressed by entries of the infinite matrix $\bU_n$.

We first derive the dynamical relations of $\bU_n$ based on the objects introduced in section \ref{S:DL}.
These relations in a sense form the infinite matrix representation of the resulting \ac{3D} semi-discrete equations.
The plane wave factors in \eqref{A:PWF} imply that the infinite matrix $\bC_n$ satisfies the dynamics as follows:
\bse\label{A:CDyn}
\begin{align}
&\partial_1\bC_n=\bLd\bC_n+\bC_n \tbLd, \label{A:CDyna} \\
&\wt\bC_n\frac{p_1-\tbLd}{p_1+\tbLd}=\frac{p_1+\bLd}{p_1-\bLd}\bC_n, \label{A:CDynb} \\
&\bC_{n+1}(-\tbLd)=\bLd\bC_n. \label{A:CDync}
\end{align}
\ese
These equations follow from $\bC_n$ through direct computation.
Meanwhile, we are able to deduce from \eqref{A:Kernel} that
\begin{align}\label{A:OaDyn}
\bOa\bLd+\tbLd\bOa=\bO.
\end{align}
Equations \eqref{A:CDyn} and \eqref{A:OaDyn} together determine dynamical evolutions of the infinite matrix $\bU_n$.
We list them in the following proposition.
\begin{proposition}
The infinite matrix $\bU_n$ satisfies dynamical evolutions as follows:
\bse\label{A:UDyn}
\begin{align}
&\partial_1\bU_n=\bLd\bU_n+\bU_n\tbLd-\bU_n\bO\bU_n, \label{A:UDyna} \\
&\wt\bU_n\frac{p_1-\tbLd}{p_1+\bLd}=\frac{p_1+\bLd}{p_1-\bLd}\bU_n-2p_1\wt\bU_n\frac{1}{p_1+\tbLd}\bO\frac{1}{p_1-\bLd}\bU_n, \label{A:UDynb} \\
&\bU_{n+1}(-\tbLd)=\bLd\bU_n-\bU_{n+1}\bO\bU_n. \label{A:UDync}
\end{align}
\ese
\end{proposition}
\begin{proof}
See appendix \ref{A:UDynPf}.
\end{proof}

The significance of system \eqref{A:UDyn} is that it establishes the connection
between the dynamical evolutions of $\bU_n$ and the purely algebraic operations involving $\bU_n$, $\bLd$, $\tbLd$ and $\bO$.
These are important relations to construct closed-form semi-discrete equations based on the infinite matrix system \eqref{A:UDyn},
through identities in terms of the entries of $\bU_n$.
For convenience, we introduce the following new variables based on the entries of the infinite matrix $\bU_n$ as follows:
\begin{align*}
&u_n=\bU_n^{(0,0)}, \quad v_n=1-\left(\bU_n\tbLd^{-1}\right)^{(0,0)}, \quad w_n=1-\left(\bLd^{-1}\bU_n\right)^{(0,0)}, \\
&V_n(a)=1-\left(\bU_n\frac{1}{a+\tbLd}\right)^{(0,0)}, \quad W_n(a)=1-\left(\frac{1}{a+\bLd}\bU_n\right)^{(0,0)}, \quad S_n(a,b)=\left(\frac{1}{a+\bLd}\bU_n\frac{1}{b+\tbLd}\right)^{(0,0)}.
\end{align*}
We note that in the definitions of $V_n(a)$, $W_n(a)$ and $S_n(a,b)$ the fractional linear expression of $\bLd$ and $\tbLd$ should be understood as their respective formal series expansions, namely these variables are determined by an infinite number of entries of $\bU_n$.

In our algebraic construction, a very important closed-form equation is the one expressed by the tau function.
For this reason, we first derive possible dynamical evolutions of the tau function \eqref{tau} with respect to the independent variables.
In the following proposition, we list the most fundamental ones.
\begin{proposition}
The tau function satisfies dynamical evolutions as follows:
\bse\label{A:tauDyn}
\begin{align}
&\partial_1\ln\tau_n=u_n, \label{A:tauDyna} \\
&\frac{\wt\tau_n}{\tau_n}=1+2p_1S_n(-p_1,-p_1), \quad \frac{\ut\tau_n}{\tau_n}=1-2p_1S_n(p_1,p_1), \label{A:tauDynb} \\
&\frac{\tau_{n+1}}{\tau_n}=v_n \quad \hbox{and} \quad \frac{\tau_{n-1}}{\tau_n}=w_n. \label{A:tauDync}
\end{align}
\ese
\end{proposition}
\begin{proof}
See appendix \ref{A:tauDynPf}.
\end{proof}
The relations listed in \eqref{A:tauDyn} establish the connection between the tau function and the new variables.
Then the idea is to search for a relation between the new variables based on the dynamics of $\bU_n$ given by \eqref{A:UDyn},
and consequently to construct a closed-form equation of the tau function.
We conclude the result as the following theorem.

\begin{theorem}
The tau function defined by \eqref{tau} satisfies the \ac{DDeltaE}
\begin{align}\label{A:tau}
\partial_1\ln\left(2p_1+\partial_1\ln\frac{\tau_n}{\wt\tau_n}\right)
=\partial_1\ln\frac{\tau_{n+1}\tau_{n-1}}{\tau_n^2}
+\left(2p_1+\partial_1\ln\frac{\tau_n}{\wt\tau_n}\right)
-\frac{\displaystyle 2p_1+\partial_1\ln\frac{\tau_{n+1}}{\wt\tau_{n+1}}}{\displaystyle 1+\frac{\wt{\tau}_n\tau_{n+1}}{\tau_n\wt{\tau}_{n+1}}}
-\frac{\displaystyle 2p_1+\partial_1\ln\frac{\tau_{n-1}}{\wt\tau_{n-1}}}{\displaystyle 1+\frac{\wt{\tau}_n\tau_{n-1}}{\tau_n\wt{\tau}_{n-1}}}.
\end{align}
\end{theorem}
\begin{proof}
See appendix \ref{A:tauPf}.
\end{proof}
Equation \eqref{A:tau} is effectively a homogenous equation of degree 6, and cannot be written as a scalar bilinear equation in Hirota's form.
In addition to the scalar form expressed by the tau function, i.e. \eqref{A:tau}, we are also interested in other nonlinear forms.
For example, we are able to derive a $(2+1)$-dimensional semi-discrete equation expressed by the potential $u_n=\bU_n^{(0,0)}$ from \eqref{A:tau}, with the help of \eqref{A:tauDyna}.
This implies that we have the following theorem.
\begin{theorem}
For arbitrary solution $\bu_n(k)$ to the linear integral equation \eqref{Integral}, the variable $u_n=\bU_n^{(0,0)}$ following from the infinite potential matrix \eqref{Potential} provides a solution of the semi-discrete equation
\begin{align}\label{A:u}
\partial_1\ln(2p_1+u_n-\wt u_n)={}&u_{n+1}-2u_n+u_{n-1}+\left(2p_1+u_n-\wt u_n\right) \nonumber \\
&-\frac{2p_1+u_{n+1}-\wt u_{n+1}}{1+\exp\left[\partial_1^{-1}(u_{n+1}-\wt u_{n+1})-\partial_1^{-1}(u_n-\wt u_n)\right]}
-\frac{2p_1+u_{n-1}-\wt u_{n-1}}{1+\exp\left[\partial_1^{-1}(u_{n-1}-\wt u_{n-1})-\partial_1^{-1}(u_n-\wt u_n)\right]}.
\end{align}
\end{theorem}
\begin{remark}
Alternatively, for arbitrary solution $\tbv_n(k')$ of the linear integral equation \eqref{AdIntegral}, the variable $u_n=\bU_n^{(0,0)}$ following from \eqref{AdPotential} provides a solution of \eqref{A:u}. This is a parallel result to the above theorem.
\end{remark}
To the best of the authors' knowledge, equation \eqref{A:u} is a novel semi-discrete equations that has not appeared in the literature.
Equation \eqref{A:u} is nonlocal in the sense that it involves an integration with respect to the continuous independent variable $x_1$.
We can, of course, take $\ln\tau_n$ as the nonlinear potential instead of $\partial_1\ln\tau_n$.
As a result, we are able to obtain from \eqref{A:tau} a \ac{DDeltaE} which is second-order in $x_1$ and $n$ and first-order in $n_1$.
However, we still prefer the potential $u_n$, as it is widely adopted in the theory of continuous integrable systems.
Equation \eqref{A:u} is the `standard' (versus the notion of `modified') equation in its potential form, as the potential $u_n$ is the first-order derivative of the logarithm of $\tau_n$.
By introducing a new variable $U_n\doteq\wt u_n-u_n$, we are able to derive
\begin{align}\label{A:U}
\partial_1\ln\left(2p_1-\wt U_n\right)-\partial_1\ln\left(2p_1-U_n\right)={}&U_{n+1}-2U_n+U_{n-1}+(U_n-\wt U_n) \nonumber \\
&+\frac{2p_1-U_{n+1}}{1+\exp\left[\partial_1^{-1}U_{n}-\partial_1^{-1}U_{n+1}\right]}-\frac{2p_1-\wt U_{n+1}}{1+\exp\left[\partial_1^{-1}\wt U_{n}-\partial_1^{-1}\wt U_{n+1}\right]} \nonumber \\
&+\frac{2p_1-U_{n-1}}{1+\exp\left[\partial_1^{-1}U_{n}-\partial_1^{-1}U_{n-1}\right]}-\frac{2p_1-\wt U_{n-1}}{1+\exp\left[\partial_1^{-1}\wt U_{n}-\partial_1^{-1}\wt U_{n-1}\right]},
\end{align}
which we refer to as the nonpotential form of \eqref{A:u}.
\begin{remark}
In the continuous theory, there also exist modified equations expressed by $v_n$ and $w_n$.
Notice that there is a transform \eqref{A:tauDync} between $v_n$ and $\tau_n$.
We can conclude that for arbitrary solution $\bu_n(k)$ of \eqref{Integral}, the variable $v_n=1-\bU_n^{(0,-1)}$ is a solution to the semi-discrete equation
\begin{align}\label{A:v}
\partial_1\ln\left(2p_1+\partial_1\Delta^{-1}\ln\frac{v_n}{\wt v_n}\right)=&{}\partial_1\ln\frac{v_n}{v_{n-1}}+\left(2p_1+\partial_1\Delta^{-1}\ln\frac{v_n}{\wt v_n}\right)
-\frac{\displaystyle 2p_1+\partial_1\Delta^{-1}\ln\frac{v_{n+1}}{\wt v_{n+1}}}{\displaystyle 1+\frac{v_n}{\wt v_n}}
-\frac{\displaystyle 2p_1+\partial_1\Delta^{-1}\ln\frac{v_{n-1}}{\wt v_{n-1}}}{\displaystyle 1+\frac{\wt v_{n-1}}{v_{n-1}}},
\end{align}
where $\Delta$ is a difference operator defined by $\Delta v_n\doteq v_{n+1}-v_n$, and $\Delta^{-1}$ is its inverse satisfying $\Delta\Delta^{-1}=\Delta^{-1}\Delta=\mathrm{id}$.
Similarly, a semi-discrete equation for $w_n=1-\bU_n^{(-1,0)}$ can also be derived through the identity $v_nw_{n+1}=1$.
We would also like to comment that our attention will not be paid to these modified equations in the sections below,
mainly because the nonlocality in terms of $n$ in equation \eqref{A:v} will raise divergence when we perform reductions towards the semi-discrete Drinfel'd--Sokolov hierarchies.
\end{remark}

\subsection{Associated linear problems}

The proposed equation \eqref{A:u} is integrable in the sense that it possesses two Lax pairs (i.e. the Lax pair and its adjoint).
These are constructed based on the wave functions $\bu_n(k)$ and $\tbv_n(k')$ within the infinite matrix scheme.
\begin{proposition}
The wave function $\bu_n(k)$ satisfies dynamical evolutions as follows:
\bse\label{A:uDyn}
\begin{align}
&\partial_1\bu_n(k)=\bLd\bu_n(k)-\bU_n\bO\bu_n(k), \label{A:uDyna} \\
&\wt\bu_n(k)=\frac{p_1+\bLd}{p_1-\bLd}\bu_n(k)-2p_1\wt\bU_n\frac{1}{p_1+\tbLd}\bO\frac{1}{p_1-\bLd}\bu_n(k), \label{A:uDynb} \\
&\bu_{n+1}(k)=\bLd\bu_n(k)-\bU_{n+1}\bO\bu_n(k). \label{A:uDync}
\end{align}
\ese
\end{proposition}
\begin{proof}
See appendix \ref{A:uDynPf}.
\end{proof}
The equations listed in \eqref{A:uDyn} form an infinite vector system towards the linear problem for the semi-discrete equation \eqref{A:u}.
To precisely construct the associated linear problem in scalar form, we introduce a scalar wave function $\phi_n\doteq[\bu_n(k)]^{(0)}$.
We present the result as the theorem below.
\begin{theorem}
Suppose that $\bu_n(k)$ is an arbitrary solution to the linear integral equation \eqref{Integral} and $\bU_n$ is the corresponding potential matrix defined by \eqref{Potential}.
The scalar wave function $\phi_n=[\bu_n(k)]^{(0)}$ and the potential $u_n=\bU_n^{(0,0)}$ satisfy the linear system
\bse\label{A:Lax}
\begin{align}
&\partial_1\phi_n=\phi_{n+1}+(u_{n+1}-u_n)\phi_n, \label{A:Laxa} \\
&\wt{\phi}_{n+1}=\frac{2p_1+u_{n+1}-\wt u_{n+1}}{1+\exp\left[\partial_1^{-1}(\wt u_{n+1}-u_{n+1})-\partial_1^{-1}(\wt u_n-u_n)\right]}\wt{\phi}_n
-\phi_{n+1}-\frac{2p_1+u_{n+1}-\wt u_{n+1}}{1+\exp\left[\partial_1^{-1}(u_{n+1}-\wt u_{n+1})-\partial_1^{-1}(u_n-\wt u_n)\right]}\phi_n. \label{A:Laxb}
\end{align}
\ese
\end{theorem}
\begin{proof}
See appendix \ref{A:LaxPf}.
\end{proof}
\begin{remark}
The linear equation \eqref{A:Laxa} is exactly the same as the continuous part of the Lax pair of \eqref{DDCKP} by taking lattice parameter zero, see e.g. \cite{NP94},
which is also an alternative representation of the spectral problem for the \ac{2D} Toda system, cf. \cite{FG80,FG83}.
Equation \eqref{A:Laxb} is, however, a novel one that has not yet appeared in the literature, as far as we know.
\end{remark}
The linear system \eqref{A:Lax} forms a Lax pair for the semi-discrete equation \eqref{A:u},
as the compatibility condition
\begin{align*}
\partial_1(\wt\phi_{n+1})=(\partial_1\phi)_{n+1}^{\wt{\,}}
\end{align*}
gives rise to the semi-discrete equation \eqref{A:u}.
The linear system \eqref{A:Lax} also serves as the Lax pair for equations \eqref{A:tau} or \eqref{A:v},
once the potential $u_n$ is replaced by $\tau_n$ or $v_n$.
To construct the adjoint Lax pair of equation \eqref{A:u}, we need to focus on the wave function $\tbv_n(k')$.
The main results are presented in the proposition and theorem below.
\begin{proposition}
The wave function $\tbv_n(k')$ satisfies dynamical evolutions as follows:
\bse\label{A:vDyn}
\begin{align}
&\partial_1\tbv_n({k'})=\tbv_n({k'})\tbLd-\tbv_n({k'})\bO\bU_n, \label{A:vDyna}\\
&\ut\tbv_n(k')=\tbv_n(k')\frac{p_1-\tbLd}{p_1+\tbLd}+2p_1\tbv_n(k')\frac{1}{p_1+\tbLd}\bO\frac{1}{p_1-\bLd}\ut\bU_n, \label{A:vDynb}\\
&\tbv_{n-1}(k')=-\tbv_n(k')\tbLd+\tbv_n(k')\bO\bU_{n-1}. \label{A:vDync}
\end{align}
\ese
\end{proposition}

\begin{theorem}
For an arbitrary solution $\tbv_n(k')$ to the linear integral equation \eqref{AdIntegral} and associated $\bU_n$ given by \eqref{AdPotential},
the adjoint scalar wave function $\psi_n=[\tbv_n(k')]^{(0)}$ and the nonlinear potential $u_n=\bU_n^{(0,0)}$ satisfy the linear system
\bse\label{A:AdLax}
\begin{align}
&\partial_1\psi_n=-\psi_{n-1}+(u_{n-1}-u_n)\psi_n, \label{A:AdLaxa} \\
&\ut{\psi}_{n-1}=\frac{2p_1+\ut{u}_{n-1}-u_{n-1}}{1+\exp\left[\partial_1^{-1}(\ut{u}_{n-1}-u_{n-1})-\partial_1^{-1}(\ut{u}_n-u_n)\right]}\ut{\psi}_n
-\psi_{n-1}-\frac{2p_1+\ut{u}_{n-1}-u_{n-1}}{1+\exp\left[\partial_1^{-1}(u_{n-1}-\ut{u}_{n-1})-\partial_1^{-1}(u_n-\ut{u}_n)\right]}\psi_n. \label{A:AdLaxb}
\end{align}
\ese
\end{theorem}
The linear system \eqref{A:AdLax} forms the adjoint Lax pair of the semi-discrete equation \eqref{A:u}.
We note that such a linear problem is not gauge equivalent to \eqref{A:Lax}.

\subsection{Continuum limit}

We present two different continuum limit schemes to show that the proposed new semi-discrete \eqref{A:u} is the integrable semi-discretisation of both the differential-difference \ac{KP} equation and the \ac{2D} Toda equation.

We first let $n_1\rightarrow\infty$ and $p_1\rightarrow\infty$ and introduce the change of variables $(x_1,n_1,n)\rightarrow(x_1,x_3,n)$ given by
\begin{align}\label{Lim1}
x_1+\frac{2n_1}{p_1}\doteq x_1, \quad \frac{2n_1}{3p_1^3}\doteq x_3 \quad \hbox{and} \quad n \doteq n.
\end{align}
Then by series expansion of \eqref{A:tau}, we obtain a \ac{DDeltaE}
\begin{align}\label{A:tauLim1}
\partial_3\ln\frac{\tau_{n+1}\tau_{n-1}}{\tau_{n}^2}=\partial_1^3\ln\left(\tau_{n+1}\tau_{n}\tau_{n-1}\right)
+3\left(\partial_1^2\ln\tau_{n+1}\right)\left(\partial_1\ln\frac{\tau_{n+1}}{\tau_{n}}\right)
-3\left(\partial_1^2\ln\tau_{n-1}\right)\left(\partial_1\ln\frac{\tau_{n}}{\tau_{n-1}}\right)
+\left(\partial_1\ln\frac{\tau_{n+1}}{\tau_{n}}\right)^3
-\left(\partial_1\ln\frac{\tau_{n}}{\tau_{n-1}}\right)^3
\end{align}
arising as the coefficient of the lowest order term, i.e. the term of $O(p_1^{-2})$.
Equation \eqref{A:tauLim1} is equivalent to the coupled system of bilinear equations (see the appendix of \cite{JM83}) as follows:
\begin{align*}
(\rD_1^2+\rD_2)\tau_n\cdot\tau_{n+1}=0, \quad (\rD_1^3-4\rD_3-3\rD_1\rD_2)\tau_n\cdot\tau_{n+1}=0,
\end{align*}
by eliminating the derivative with respect to $x_2$, where $\rD_j$ stands for Hirota's bilinear operator defined as
\begin{align*}
\rD_j f\cdot g=\left.\left(\frac{\partial}{\partial x_j}-\frac{\partial}{\partial x'_j}\right)\right|_{x'_j=x_j}f(\cdots,x_j,\cdots)g(\cdots,x'_j,\cdots)
\end{align*}
for arbitrary differentiable functions $f=f(x_1,x_2,\cdots)$ and $g=g(x_1,x_2,\cdots)$.
Similarly, we can also perform \eqref{Lim1} on the nonlinear equation \eqref{A:u} and evaluate its continuum limit. As a result, we obtain a \ac{DDeltaE} in the form of
\begin{align}\label{A:uLim1}
\partial_3(u_{n+1}-2u_n+u_{n-1})=\partial_1^3(u_{n+1}+u_n+u_{n-1})
+\partial_1\left[3(u_{n+1}-u_n)\partial_1u_{n+1}-3(u_n-u_{n-1})\partial_1u_{n-1}+(u_{n+1}-u_n)^3-(u_n-u_{n-1})^3\right].
\end{align}
We remark that equation \eqref{A:uLim1} is the potential form of the third-order differential-difference \ac{KP} equation introduced in \cite{Fu13}.

Alternatively, we let $n_1\rightarrow\infty$ and $p_1\rightarrow 0$ and introduce new coordinates $(x_1,x_{-1},n)$ which are connected with the old ones $(x_1,n_1,n)$ through
\begin{align}\label{Lim2}
x_1\doteq x_1, \quad 2n_1p_1\doteq x_{-1} \quad \hbox{and} \quad n\doteq n.
\end{align}
Expanding equation \eqref{A:tau} into a series in terms of the positive powers of $p_1$, we obtain from the leading term the following equation:
\begin{align}\label{A:tauLim2}
\frac{1}{2}\rD_1\rD_{-1}\tau_n\cdot\tau_n=\tau_n^2-\tau_{n+1}\tau_{n-1},
\end{align}
namely the bilinear \ac{2D} Toda equation, see e.g. \cite{UT84,Hir04}.
The change of coordinates \eqref{Lim2} also brings us another \ac{DDeltaE}
\begin{align}\label{A:uLim2}
\partial_1\ln(1-\partial_{-1}u_n)=u_{n+1}-2u_n+u_{n-1}
\end{align}
as the continuum limit of \eqref{A:u}. Equation \eqref{A:uLim2} is one of the nonlinear forms of the \ac{2D} Toda equation (see e.g. \cite{Fu18a}),
gauge equivalent to the well-known form
\begin{align*}
\partial_1\partial_{-1}\varphi_n=\re^{\varphi_n-\varphi_{n-1}}-\re^{\varphi_{n+1}-\varphi_n},
\end{align*}
given by Mikhailov \cite{Mik79}, through the Miura-type transformation $u_{n+1}-u_n=\partial_1\varphi_n$.

\section{Semi-discrete Drinfel'd--Sokolov hierarchies}\label{S:DS}

\subsection{General reduction formulae}

The semi-discrete equation \eqref{A:u} is the one associated with the algebra $A_\infty$, as the most general model in this paper.
In this section, we construct the semi-discrete equations associated with infinite-dimensional algebras $B_\infty$ and $C_\infty$ and also Kac--Moody algebras $A_r^{(1)}$, $A_{2r}^{(2)}$, $C_r^{(1)}$ and $D_{r+1}^{(2)}$ in the Drinfel'd--Sokolov classification, from the semi-discrete equation \eqref{A:u}.
This is realised by imposing restrictions on the integration measure $\rd\zeta(k,k')$ and the integration domain $D$ (see\cite{Fu21b}),
which leads to symmetry and periodicity constraints on the discrete independent variable $n$.
As a consequence, the independent variable $n$ turns out to be an index, labelling multi-component variables for the reduced integrable equations.

We start with the reduction to $B_\infty$, by taking a symmetric domain $D$ and a special measure $\rd\zeta(k,k')$ satisfying
\begin{align}\label{B:Measure}
\rd\zeta(k,k')=\rd\zeta'(k,k')k, \quad \hbox{in which} \quad \rd\zeta'(k,k')=-\rd\zeta'(k',k), \quad \forall (k,k')\in D.
\end{align}
Notice that the infinite matrix $C_n$ defined as \eqref{C} with the plane wave factors \eqref{A:PWF} reads
\begin{align*}
\bC_n=\iint_D\rd\zeta'(k,k')k\bc(k)\tbc(k')\re^{(k+k')x_{1}}\left(\frac{p_1+k}{p_1-k}\frac{p_1+k'}{p_1-k'}\right)^{n_1}\left(-\frac{k}{k'} \right)^n
\end{align*}
in the $A_\infty$ case.
It is direct to verify that under the special constraint \eqref{B:Measure} we have
\begin{align*}
\tbC_n&=\iint_D\rd\zeta'(k,k')k\bc(k')\tbc(k)\re^{(k+k')x_{1}}\left(\frac{p_1+k}{p_1-k}\frac{p_1+k'}{p_1-k'}\right)^{n_1}\left(-\frac{k}{k'} \right)^n \\
&=\iint_D\rd\zeta'(k',k)k'\bc(k')\tbc(k)\re^{(k'+k)x_{1}}\left(\frac{p_1+k'}{p_1-k'}\frac{p_1+k}{p_1-k}\right)^{n_1}\left(-\frac{k'}{k}\right)^{-1-n}=\bC_{-1-n}.
\end{align*}
This together with the symmetry $\tbOa=\bOa$ results in the same reduction on $\bU_n$ and $\tau_n$, by doing the same analysis on \eqref{U} and \eqref{tau}.
In other words, we have the symmetry reductions
\begin{align}\label{B:Reduc}
\tau_{-1-n}=\tau_n \quad \hbox{and} \quad u_{-1-n}=u_n
\end{align}
in the $B_\infty$ case.
Similarly, we impose the restriction that $D$ and $\rd\zeta(k,k')$ are both symmetric for the $C_\infty$ class, namely
\begin{align}\label{C:Measure}
\rd\zeta(k,k')=\rd\zeta(k',k), \quad \forall (k,k')\in D.
\end{align}
Thus we accordingly obtain the reduction conditions
\begin{align}\label{C:Reduc}
\tau_{-n}=\tau_{n} \quad \hbox{and} \quad u_{-n}=u_n,
\end{align}
which help to reduce \eqref{A:tau} and \eqref{A:u} to the $C_\infty$-type equations, respectively.
We note that in these two classes there is no reduction on the wave functions $\phi_n$ and $\psi_n$.
Hereby, the corresponding Lax pairs share the same form of \eqref{A:Lax} and \eqref{A:AdLax}, subject to \eqref{B:Reduc} and \eqref{C:Reduc}, respectively.

For the $A_r^{(1)}$ class, we take a special measure of the form
\begin{align}\label{Ar1:Measure}
\rd\zeta(k,k')=\sum_{j=1}^{\varphiup(\cN)}\frac{1}{2\pi\ri}\frac{\rd\lambda_j(k)\rd k'}{k'+\oa_\cN^{(j)}k},
\end{align}
in which $\oa_{\cN}^{(j)}$ for $j=1,2,\cdots,\varphiup(\cN)$, with $\varphiup(\cdot)$ being Euler's totient function, denote all the primitive $\cN$th roots of unity, and $\rd\lambda_j(k)$ are arbitrary measures.
Such a measure deep down indicates that a constraint $k^\cN=(-k')^\cN$ is imposed on the spectral parameters $k$ and $k'$.
This results in the periodicity conditions
\begin{align}\label{Period}
\bC_{n+\cN}=\bC_n \quad \hbox{and} \quad \bU_{n+\cN}=\bU_n,
\end{align}
according to \eqref{C} and \eqref{U}.
Following the definitions of $u_n$ and $\tau_n$ as well as \eqref{u} and \eqref{v}, we can induce the $A_r^{(1)}$ reduction by setting $\cN=r+1$, which is composed of
\begin{align}\label{Ar1:Reduc}
\tau_{n+r+1}=\tau_n, \quad u_{n+r+1}=u_n, \quad \phi_{n+r+1}=k^{r+1}\phi_n \quad \hbox{and} \quad \psi_{n+r+1}=k^{-(r+1)}\psi_n.
\end{align}
The formula \eqref{Ar1:Reduc} reduces the semi-discrete equation \eqref{A:u} to the $(1+1)$-dimensional differential-difference system with respect to $n_1$ and $x_1$.
However, the obtained integrable system is different from that discussed in \cite{Fu18a}.
This is because here the discrete dispersion relation described by $n_1$ relies on the spectral parameter in a fractionally linear way according to \eqref{A:PWF}.
In other words, we present a different integrable semi-discretisation of the \ac{GD} hierarchy,
which, of course, takes a more complex form compared with the existing result.

Next, we take a special measure
\begin{align}\label{BN:Measure}
\rd\zeta(k,k')=k\sum_{j=1}^{\varphiup(\cN)}\frac{1}{2\pi\ri}\left(\frac{\rd\lambda_j(k)\rd k'}{k'+\oa_{\cN}^{(j)}k}-\frac{\rd k\rd\lambda_j(k')}{k+\oa_{\cN}^{(j)}k'}\right).
\end{align}
When $\cN=2r+1$, this leads to the $A_{2r}^{(2)}$ reduction, composed of
\begin{align}\label{A2r2:Reduc}
\tau_{n+2r+1}=\tau_n, \quad \tau_{-1-n}=\tau_n, \quad u_{n+2r+1}=u_n, \quad u_{-1-n}=u_n, \quad
\phi_{n+2r+1}=k^{2r+1}\phi_n \quad \hbox{and} \quad \psi_{n+2r+1}=k^{-(2r+1)}\psi_n;
\end{align}
while for $\cN=2r+2$, we obtain the reduction to $D_{r+1}^{(2)}$, namely
\begin{align}\label{Dr2:Reduc}
\tau_{n+2r+2}=\tau_n, \quad \tau_{-1-n}=\tau_n, \quad u_{n+2r+2}=u_n, \quad u_{-1-n}=u_n, \quad
\phi_{n+2r+2}=k^{2r+2}\phi_n \quad \hbox{and} \quad \psi_{n+2r+2}=k^{-(2r+2)}\psi_n.
\end{align}

Finally, we consider a special measure in the form of
\begin{align}\label{CN:Measure}
\rd\zeta(k,k')=\sum_{j=1}^{\varphiup(\cN)}\frac{1}{2\pi\ri}\left(\frac{\rd\lambda_j(k)\rd k'}{k'+\oa_{\cN}^{(j)}k}+\frac{\rd k\,\rd\lambda_j(k')}{k+\oa_{\cN}^{(j)}k'}\right).
\end{align}
In the case of $\cN=2r$, we derive the constraints
\begin{align}\label{Cr1:Reduc}
\tau_{n+2r}=\tau_n, \quad \tau_{-n}=\tau_n, \quad u_{n+2r}=u_n, \quad u_{-n}=u_n, \quad
\phi_{n+2r}=k^{2r}\phi_n \quad \hbox{and} \quad \psi_{n+2r}=k^{-2r}\psi_n,
\end{align}
which is the $C_r^{(1)}$ reduction.
The $\cN=2r+1$ case gives rise to
\begin{align*}
\tau_{n+2r+1}=\tau_n, \quad \tau_{-n}=\tau_n, \quad u_{n+2r+1}=u_n, \quad u_{-n}=u_n, \quad
\phi_{n+2r+1}=k^{2r+1}\phi_n \quad \hbox{and} \quad \psi_{n+2r+1}=k^{-(2r+1)}\psi_n,
\end{align*}
which is equivalent to the $A_{2r}^{(2)}$ reduction.

In addition to the reduction conditions associated with $A_r^{(1)}$, $A_{2r}^{(2)}$, $C_r^{(1)}$ and $D_{r+1}^{(2)}$,
there also exist two constraints that are applicable to each class of the reduced semi-discrete equations due to the periodicity, given by
\begin{align}\label{Constraint}
\partial_1\sum_{n=0}^{\cN-1}u_n=\sum_{n=0}^{\cN-1}(u_{n+1}-u_n)u_n \quad \hbox{and} \quad
\prod_{n=0}^{\cN-1}\frac{2p_1+u_{n+1}-\wt u_{n+1}}{1+\exp\left[\partial_1^{-1}(u_{n+1}-\wt u_{n+1})-\partial_1^{-1}(u_n-\wt u_n)\right]}=p_1^\cN,
\end{align}
where $u_{n+\cN}=u_n$. The first one is nothing but the $(0,0)$-entry of the identity
\begin{align*}
\partial_1\sum_{n=0}^{\cN-1}\bU_n=\sum_{n=0}^{\cN-1}(\bLd\bU_n+\bU_n\tbLd-\bU_n\bO\bU_n)
=\sum_{n=0}^{\cN-1}(\bLd\bU_n+\bU_{n+1}\tbLd-\bU_n\bO\bU_n)=\sum_{n=0}^{\cN-1}(\bU_{n+1}\bO\bU_n-\bU_n\bO\bU_n),
\end{align*}
which follows from \eqref{A:UDyna}, \eqref{A:UDync} and \eqref{Period}.
The second one is a direct consequence of
\begin{align*}
\prod_{n=0}^{\cN-1}\frac{V_{n+1}(-p_1)}{V_n(-p_1)}
=\prod_{n=0}^{\cN-1}\frac{\displaystyle 2p_1+\partial_1\ln\frac{\tau_{n+1}}{\wt\tau_{n+1}}}{\displaystyle p_1\left(1+\frac{\wt{\tau}_n\tau_{n+1}}{\tau_n\wt{\tau}_{n+1}}\right)}=1
\end{align*}
due to \eqref{A:tauVWb} and \eqref{Period} as well as \eqref{A:tauDyna}.

The reduced semi-discrete equations are obtained by imposing their respective reduction conditions as well as the constraints \eqref{Constraint} on the corresponding objects in the $A_\infty$ case.

\subsection{Examples}

In this subsection, we list the semi-discrete equations associated with the Kac--Moody algebras $A_1^{(1)}$, $A_2^{(2)}$, $C_2^{(1)}$ and $D_3^{(2)}$.
Each equation is integrable in the sense of having a Lax pair taking the form of
\bse\label{Lax}
\begin{align}
&\partial_1\Phi=\bP\Phi, \label{Laxa} \\
&\bQ_1\wt\Phi=\bQ_2\Phi, \label{Laxb}
\end{align}
\ese
where $\Phi={}^{t\!}(\phi_0,\phi_1,\cdots,\phi_{\cN-1})$, and $\bP$, $\bQ_1$ and $\bQ_2$, as we shall see below, are $\mathbb{Z}_\cN$ graded matrices.
The linear equation \eqref{Laxb} is a new discretisation of the Lax scheme discussed in \cite{FG80,FG83}.
This is because the effective Lax matrix in \eqref{Laxb}, i.e. $\bQ_1^{-1}\bQ_2$, takes the form of a fraction of $\mathbb{Z}_\cN$ graded matrices,
which differs from the one for the discrete Bogoyavlensky-type equations, cf. \cite{PN96,FX17a}.
We also note that on the \ac{2D} level the adjoint Lax pair is gauge equivalent to the standard one, in contrast to the \ac{3D} case.
For this reason, we do not list the adjoint ones in this subsection.
In addition, we adopt a new notation
\begin{align*}
e_{m,n}\doteq 1+\exp\left[\partial_1^{-1}(u_m-\wt u_m)-\partial_1^{-1}(u_n-\wt u_n)\right]
\end{align*}
for examples below, in order to express our results more compactly.

\begin{example}
The $A_1^{(1)}$-type equation is a two-component system composed of
\bse\label{A11:u}
\begin{align}
\partial_1\ln(2p_1+u_0-\wt u_0)&=2u_1-2u_0+(2p_1+u_0-\wt u_0)-\frac{2(2p_1+u_1-\wt u_1)}{e_{1,0}}, \\
\partial_1\ln(2p_1+u_1-\wt u_1)&=2u_0-2u_1+(2p_1+u_1-\wt u_1)-\frac{2(2p_1+u_0-\wt u_0)}{e_{0,1}}.
\end{align}
\ese
The Lax matrices $\bP$, $\bQ_1$ and $\bQ_2$ are given as follows:
\begin{align}\label{A11:Lax}
\bP=
\begin{pmatrix}
u_1-u_0 & 1 \\
k^2 & u_0-u_1
\end{pmatrix}, \quad
\bQ_1=
\begin{pmatrix}
\displaystyle\frac{2p_1+u_1-\wt u_1}{e_{0,1}} & -1  \\
-k^2 & \displaystyle\frac{2p_1+u_0-\wt u_0}{e_{1,0}}
\end{pmatrix}, \quad
\bQ_2=
\begin{pmatrix}
\displaystyle\frac{2p_1+u_1-\wt u_1}{e_{1,0}} & 1  \\
k^2 & \displaystyle\frac{2p_1+u_0-\wt u_0}{e_{0,1}}
\end{pmatrix}.
\end{align}
In this class, the potentials satisfy additional constraints
\begin{align}\label{A11:Constraint}
\partial_1(u_0+u_1)=-(u_0-u_1)^2 \quad \hbox{and} \quad
\frac{(2p_1+u_1-\wt u_1)(2p_1+u_0-\wt u_0)}{e_{0,1}e_{1,0}}=p_1^2.
\end{align}
The system \eqref{A11:u} is a new semi-discretisation of the potential \ac{KdV} equation.
Compared with the well-known differential-difference \ac{KdV} equation (see e.g. \cite{HJN16})
\begin{align*}
\partial_1(\wt u+u)=-(\wt u-u)^2+2p_1(\wt u-u),
\end{align*}
equation \eqref{A11:u} has a different $n_1$-part in the Lax pair, cf. \cite{FX17a,Fu18a}.
This is mainly caused by the discrete factor in terms of $n_1$ in the plane wave factor \eqref{A:PWF}.
\end{example}

\begin{example}
The $A_2^{(2)}$-type equation is also a two-component system in the form of
\bse\label{A22:u}
\begin{align}
2\partial_1\ln(2p_1+u_0-\wt u_0)&=2u_1-2u_0+(2p_1+u_0-\wt u_0)-\frac{2(2p_1+u_1-\wt u_1)}{e_{1,0}}, \\
\partial_1\ln(2p_1+u_1-\wt u_1)&=2u_0-2u_1+(2p_1+u_1-\wt u_1)-\frac{2(2p_1+u_0-\wt u_0)}{e_{0,1}}.
\end{align}
\ese
The corresponding Lax matrices are given by
\begin{align}\label{A22:Lax}
\bP=
\begin{pmatrix}
u_1-u_0 & 1 & 0 \\
0 & u_0-u_1 & 1 \\
k^3& 0 & 0
\end{pmatrix}, \quad
\bQ_1=
\begin{pmatrix}
a_{1,1} & -1 & 0 \\
0 & a_{2,2} & -1 \\
-k^3 & 0 & a_{3,3}
\end{pmatrix} \quad \hbox{and} \quad
\bQ_2=
\begin{pmatrix}
b_{1,1} & 1 & 0 \\
0 & b_{2,2} & 1 \\
k^3 & 0 & b_{3,3}
\end{pmatrix},
\end{align}
where the entries $a_{i,i}$ and $b_{i,i}$ are determined by
\begin{align*}
a_{1,1}=\frac{2p_1+u_1-\wt u_1}{e_{0,1}}, \quad a_{2,2}=\frac{2p_1+u_0-\wt u_0}{e_{1,0}}, \quad a_{3,3}=\frac{2p_1+u_0-\wt u_0}{2},
\end{align*}
and
\begin{align*}
b_{1,1}=\frac{2p_1+u_1-\wt u_1}{e_{1,0}}, \quad b_{2,2}=\frac{2p_1+u_0-\wt u_0}{e_{0,1}}, \quad b_{3,3}=\frac{2p_1+u_0-\wt u_0}{2},
\end{align*}
respectively.
The only difference between \eqref{A22:u} and \eqref{A11:u} is the multiplier $2$ on the left hand side of the first equation in \eqref{A22:u}.
However, we have to point out that we cannot transfer one to the other through a simple scaling, as the algebras for the two systems are entirely different,
as we can see from their respective Lax representations.
The potentials $u_0$ and $u_1$ in this class obey two additional constraints
\begin{align}\label{A22:Constraint}
\partial_1(2u_0+u_1)=-(u_0-u_1)^2 \quad \hbox{and} \quad \frac{(2p_1+u_1-\wt u_1)(2p_1+u_0-\wt u_0)^2}{2e_{0,1}e_{1,0}}=p_1^3.
\end{align}
\end{example}

\begin{example}
The $C_2^{(1)}$-type equation is a three-component system composed of the following nonlinear equations:
\bse\label{C21:u}
\begin{align}
\partial_1\ln(2p_1+u_0-\wt u_0)&=2u_1-2u_0+(2p_1+u_0-\wt u_0)-\frac{2(2p_1+u_1-\wt u_1)}{e_{1,0}}, \\
\partial_1\ln(2p_1+u_1-\wt u_1)&=u_2-2u_1+u_0+(2p_1+u_1-\wt u_1)-\frac{2p_1+u_2-\wt u_2}{e_{2,1}}-\frac{2p_1+u_0-\wt u_0}{e_{0,1}}, \\
\partial_1\ln(2p_1+u_2-\wt u_2)&=2u_1-2u_2+(2p_1+u_2-\wt u_2)-\frac{2(2p_1+u_1-\wt u_1)}{e_{1,2}}.
\end{align}
\ese
The Lax matrices in this case are given by
\begin{align}\label{C21:Lax}
\bP=
\begin{pmatrix}
u_1-u_0 & 1 & 0 & 0 \\
0 & u_2-u_1 & 1 & 0 \\
0 & 0 & u_1-u_2 & 1 \\
k^4& 0 & 0 &u_0-u_1
\end{pmatrix}, \quad
\bQ_1=
\begin{pmatrix}
a_{1,1} & -1 & 0 & 0\\
0 & a_{2,2} & -1 &0 \\
0 & 0 & a_{3,3} &-1\\
-k^4 & 0 & 0 & a_{4,4}
\end{pmatrix} \quad \hbox{and} \quad
\bQ_2=
\begin{pmatrix}
b_{1,1} & 1 & 0 & 0\\
0 & b_{2,2} & 1 &0 \\
0 & 0 & b_{3,3} & 1\\
k^4 & 0 & 0 & b_{4,4}
\end{pmatrix},
\end{align}
with the entries
\begin{align*}
a_{1,1}=\frac{2p_1+u_1-\wt u_1}{e_{0,1}}, \quad a_{2,2}=\frac{2p_1+u_2-\wt u_2}{e_{1,2}}, \quad
a_{3,3}=\frac{2p_1+u_1-\wt u_1}{e_{2,1}}, \quad a_{4,4}=\frac{2p_1+u_0-\wt u_0}{e_{1,0}},
\end{align*}
and
\begin{align*}
b_{1,1}=\frac{2p_1+u_1-\wt u_1}{e_{1,0}}, \quad b_{2,2}=\frac{2p_1+u_2-\wt u_2}{e_{2,1}}, \quad
b_{3,3}=\frac{2p_1+u_1-\wt u_1}{e_{1,2}}, \quad b_{4,4}=\frac{2p_1+u_0-\wt u_0}{e_{0,1}}.
\end{align*}
In this class, the variables $u_0$, $u_1$ and $u_2$ satisfy additional constraints
\begin{align}\label{C21:Constraint}
\partial_1(u_0+2u_1+u_2)=-(u_0-u_1)^2-(u_1-u_2)^2 \quad \hbox{and} \quad \frac{(2p_1+u_2-\wt u_2)(2p_1+u_1-\wt u_1)^2(2p_1+u_0-\wt u_0)}{e_{0,1}e_{1,2}e_{2,1}e_{1,0}}=p_1^4.
\end{align}
\end{example}

\begin{example}
The $D_3^{(2)}$-type equation is also a three-component coupled system which is composed of
\bse\label{D32:u}
\begin{align}
2\partial_1\ln(2p_1+u_0-\wt u_0)&=2u_1-2u_0+(2p_1+u_0-\wt u_0)-\frac{2(2p_1+u_1-\wt u_1)}{e_{1,0}}, \\
\partial_1\ln(2p_1+u_1-\wt u_1)&=u_2-2u_1+u_0+(2p_1+u_1-\wt u_1)-\frac{2p_1+u_2-\wt u_2}{e_{2,1}}-\frac{2p_1+u_0-\wt u_0}{e_{0,1}}, \\
2\partial_1\ln(2p_1+u_2-\wt u_2)&=2u_1-2u_2+(2p_1+u_2-\wt u_2)-\frac{2(2p_1+u_1-\wt u_1)}{e_{1,2}}.
\end{align}
\ese
The Lax matrices in this class are given by
\bse\label{D32:Lax}
\begin{align}
\bP=
\begin{pmatrix}
u_1-u_0 & 1 & 0 & 0 & 0 & 0 \\
0 & u_2-u_1 & 1 & 0 & 0 & 0 \\
0 & 0 & 0 & 1 & 0 & 0\\
0 & 0 & 0 & u_1-u_2 & 1 & 0\\
0 & 0 & 0 & 0 & u_0-u_1 & 1\\
k^6 & 0 & 0 & 0 & 0 & 0
\end{pmatrix},
\end{align}
\begin{align}
\bQ_1=
\begin{pmatrix}
a_{1,1} & -1 & 0 & 0 & 0 & 0\\
0 & a_{2,2} & -1 &0 & 0 & 0\\
0 & 0 & a_{3,3}  &-1 & 0 & 0\\
0 & 0 & 0 & a_{4,4} & -1 & 0\\
0 & 0 & 0 & 0 & a_{5,5} & -1\\
-k^6 & 0 & 0 & 0 & 0 & a_{6,6}
\end{pmatrix}\quad \hbox{and} \quad
\bQ_2=
\begin{pmatrix}
b_{1,1} & 1 & 0 & 0 & 0 & 0\\
0 & b_{2,2} & 1 &0 & 0 & 0\\
0 & 0 &  b_{3,3} &1 & 0 & 0\\
0 & 0 & 0 &  b_{4,4} & 1 & 0\\
0 & 0 & 0 & 0 & b_{5,5} & 1\\
k^6 & 0 & 0 & 0 & 0 &  b_{6,6}
\end{pmatrix},
\end{align}
\ese
where the entries $a_{i,i}$ and $b_{i,i}$ are determined by
\begin{align*}
&a_{1,1}=\frac{2p_1+u_1-\wt u_1}{e_{0,1}}, \quad a_{2,2}=\frac{2p_1+u_2-\wt u_2}{e_{1,2}}, \quad a_{3,3}=\frac{2p_1+u_2-\wt u_2}{2}, \\
&a_{4,4}=\frac{2p_1+u_1-\wt u_1}{e_{2,1}}, \quad a_{5,5}=\frac{2p_1+u_0-\wt u_0}{e_{1,0}}, \quad a_{6,6}=\frac{2p_1+u_0-\wt u_0}{2},
\end{align*}
and
\begin{align*}
&b_{1,1}=\frac{2p_1+u_1-\wt u_1}{e_{1,0}}, \quad b_{2,2}=\frac{2p_1+u_2-\wt u_2}{e_{2,1}}, \quad b_{3,3}=\frac{2p_1+u_2-\wt u_2}{2},\\
&b_{4,4}=\frac{2p_1+u_1-\wt u_1}{e_{1,2}}, \quad b_{5,5}=\frac{2p_1+u_0-\wt u_0}{e_{0,1}}, \quad b_{6,6}=\frac{2p_1+u_0-\wt u_0}{2},
\end{align*}
respectively.
Equation \eqref{D32:u} looks very similar to \eqref{C21:u}, namely the only difference occurs in the coefficients of the left hand sides of the first and third equations.
We comment that their respective algebraic structures are entirely different, from the aspects of solution and Lax representation.
Likewise, the potentials $u_0$, $u_1$ and $u_2$ in this class obey additional constraints
\begin{align}\label{D32:Constraint}
2\partial_1(u_0+u_1+u_2)=-(u_0-u_1)^2-(u_1-u_2)^2 \quad \hbox{and} \quad \frac{(2p_1+u_2-\wt u_2)^2(2p_1+u_1-\wt u_1)^2(2p_1+u_0-\wt u_0)^2}{4e_{0,1}e_{1,2}e_{2,1}e_{1,0}}=p_1^6.
\end{align}
\end{example}

The semi-discrete equations listed in this subsection are still nonlocal in terms of the flow variable $x_1$.
We can certainly localise these equations by introducing new variables such as $s_n\doteq\partial_1^{-1}u_n$.
However, in order to establish the connection between these equations and the continuous Drinfel'd--Sokolov hierarchies, we still choose $u_n$ as the potentials of these equations. Although the equations of $u_n$ are nonlocal, we shall see in the forthcoming subsection that suitable continuum limits of these equations
results in the potential forms of both positive and negative flows of the Drinfel'd--Sokolov hierarchies, which are, as expected, local equations.

We observe that all these semi-discrete equations are accompanied by additional constraints.
From the first glance, these systems look overdetermined, as in each class the number of equations is greater than that of dependent variables.
However, we would like to point out that these are actually not overdetermined systems.
In fact, both the semi-discrete equations and additional constraints are derived from the \ac{DL},
which implies that in each class the semi-discrete equations are compatible with the constraints, from the perspective of the corresponding solution space.
To put it another way, the constraints do not further restrict the solution spaces of corresponding semi-discrete equations.

From this aspect, it is reasonable to ignore the constraints
and only think of \eqref{A11:u}, \eqref{A22:u}, \eqref{C21:u} and \eqref{D32:u} as the semi-discrete Drinfel'd--Sokolov equations.
These equations, after all, arise as the compatibility conditions of their corresponding Lax pairs.
The additional constraints present restrictions on the variables $u_n$ in each class,
but unfortunately, it seems not possible to reduce the number of variables in the coupled semi-discrete systems here due to the feature of nonlocality\footnote{
This is not surprising because such an issue also occurred in the literature.
In \cite{FX17a}, the discrete \ac{GD}-type equations were also presented as multi-component systems
accompanied by additional constraints (which were referred to as first integrals).
As was pointed out by those authors, it is not always possible to reduce the number of potential by making use of the additional constraints.}.
We will see in the forthcoming subsection that the number of components can be reduced in the continuous case by using the continuous analogues of these constraints,
leading to well-known \ac{PDE}s in the Drinfel'd--Sokolov classification.

The semi-discrete equations \eqref{A11:u}, \eqref{A22:u}, \eqref{C21:u} and \eqref{D32:u} play a role of the \ac{BT}s for their corresponding continuous equations.
To be more precise, these equations are the \ac{BT}s of the multi-component systems \eqref{A11:KdV}, \eqref{A22:KdV}, \eqref{C21:KdV} and \eqref{D32:KdV}, respectively.
While the additional constraints in each class act as the non-auto \ac{BT} (like Miura transformation) between the potentials $u_n$.
For example, the \ac{SK} equation \eqref{SK} and the \ac{KK} equation \eqref{KK} are two separate scalar models in the $A_2^{(2)}$ class in the Drinfel'd--Sokolov classification.
The first equation in \eqref{A22:Constraint} is the non-auto \ac{BT} between \eqref{SK} and \eqref{KK}.
The \ac{SK} and \ac{KK} equations together can be reformulated as a two-component systems \eqref{A22:KdV}.
Then the semi-discrete equation \eqref{A22:u} plays a role of its auto \ac{BT} from $(u_0,u_1)$ to $(\wt u_0,\wt u_1)$.

\subsection{Continuum limits to the Kortweg--de Vries-type and \ac{2D} Toda-type equations}

To convince us that the obtained $(1+1)$-dimensional equations are suitable semi-discretisation of the Drinfel'd--Sokolov equations,
continuum limits towards the \ac{KdV}-type and \ac{2D} Toda-type equations are discussed in this subsection, illustrated by the above examples.

We first consider the continuum limits to the \ac{KdV}-type equations.
In the cases of $A_1^{(1)}$, $C_2^{(1)}$ and $D_3^{(2)}$,
we let $n_1\rightarrow\infty$ and $p_1\rightarrow\infty$ and introduce the change of variables $(x_1,n_1)\rightarrow(x_1,x_3)$ given by
\begin{align}\label{A11:Lim}
x_1+\frac{2n_1}{p_1}\doteq x_1 \quad \hbox{and} \quad \frac{2n_1}{3p_1^3}\doteq x_3.
\end{align}
While in the $A_2^{(2)}$ case, the reduction of period $3$ implies that
the corresponding \ac{PDE} describes evolutions with respect to the continuous independent variables $x_1$ and $x_5$.
Hence, for $A_2^{(2)}$-type we let $n_1\rightarrow\infty$ and $p_1\rightarrow\infty$ and simultaneously introduce the new coordinates $(x_1,x_5)$
which are connected with the old ones $(x_1,n_1)$ through
\begin{align}\label{A22:Lim}
x_1+\frac{2n_1}{p_1}\doteq x_1 \quad \hbox{and} \quad \frac{2n_1}{5p_1^5}\doteq x_5.
\end{align}
Then taking the respective continuum limits of \eqref{A11:u}, \eqref{A22:u}, \eqref{C21:u} and \eqref{D32:u}, we obtain continuous multi-component systems as follows.
\begin{description}
\item[$A_1^{(1)}$]
\bse\label{A11:KdV}
\begin{align}
-\partial_3(2u_0-2u_1)={}&\partial_1^3(u_0+2u_1)+\partial_1\left[6(u_1-u_0)\partial_1u_1+2(u_1-u_0)^3\right], \label{A11:KdVa} \\
-\partial_3(-2u_0+2u_1)={}&\partial_1^3(2u_0+u_1)+\partial_1\left[6(u_0-u_1)\partial_1u_0+2(u_0-u_1)^3\right], \label{A11:KdVb}
\end{align}
\ese
\item[$A_2^{(2)}$]
\bse\label{A22:KdV}
\begin{align}
-\partial_5(u_0-u_1)={}&\frac{1}{3}\partial_1^5(7u_0+3u_1)
+\frac{1}{3}\partial_1\left[15(\partial_1u_0)(\partial_1^2u_0)-15(\partial_1u_0)(\partial_1^2u_1)
-10(\partial_1u_1)(\partial_1^2u_0)+25(\partial_1u_1)(\partial_1^2u_1)\right. \nonumber \\
&\qquad\qquad\qquad\qquad\qquad+10(u_1-u_0)\partial_1^3u_1+15(u_1-u_0)(\partial_1u_1-\partial_1u_0)^2 \nonumber \\
&\qquad\qquad\qquad\qquad\qquad\left.+10(u_1-u_0)^2\partial_1^2(u_1-u_0)-5(u_1-u_0)^3\partial_1u_1-2(u_1-u_0)^5\right], \\
-\partial_5(-2u_0+2u_1)={}&\frac{1}{3}\partial_1^5(6u_0+4u_1)
+\frac{1}{3}\partial_1\left[50(\partial_1u_0)(\partial_1^2u_0)-20(\partial_1u_0)(\partial_1^2u_1)
-30(\partial_1u_1)(\partial_1^2u_0)+15(\partial_1u_1)(\partial_1^2u_1)\right. \nonumber \\
&\qquad\qquad\qquad\qquad\qquad-20(u_1-u_0)\partial_1^3u_0-30(u_1-u_0)(\partial_1u_1-\partial_1u_0)^2 \nonumber \\
&\qquad\qquad\qquad\qquad\qquad\left.-20(u_1-u_0)^2\partial_1^2(u_1-u_0)+10(u_1-u_0)^3\partial_1u_0+4(u_1-u_0)^5\right],
\end{align}
\ese
\item[$C_2^{(1)}$]
\bse\label{C21:KdV}
\begin{align}
-\partial_3(2u_0-2u_1)={}&\partial_1^3(u_0+2u_1)+\partial_1\left[6(u_1-u_0)\partial_1u_1+2(u_1-u_0)^3\right], \\
-\partial_3(-u_0+2u_1-u_2)={}&\partial_1^3(u_0+u_1+u_2)+\partial_1\left[3(u_2-u_1)\partial_1u_2-3(u_1-u_0)\partial_1u_0+(u_2-u_1)^3-(u_1-u_0)^3\right], \\
-\partial_3(-2u_1+2u_2)={}&\partial_1^3(2u_1+u_2)+\partial_1\left[6(u_1-u_2)\partial_1u_1+2(u_1-u_2)^3\right],
 \end{align}
\ese
\item[$D_3^{(2)}$]
\bse\label{D32:KdV}
\begin{align}
-\partial_3(u_0-u_1)={}&\partial_1^3(2u_0+u_1)+\partial_1\left[3(u_1-u_0)\partial_1u_1+(u_1-u_0)^3\right], \\
-\partial_3(-u_0+2u_1-u_2)={}&\partial_1^3(u_0+u_1+u_2)+\partial_1\left[3(u_2-u_1)\partial_1u_2-3(u_1-u_0)\partial_1u_0+(u_2-u_1)^3-(u_1-u_0)^3\right], \\
-\partial_3(-u_1+u_2)={}&\partial_1^3(u_1+2u_2)+\partial_1\left[3(u_1-u_2)\partial_1u_1+(u_1-u_2)^3\right].
\end{align}
\ese
\end{description}
These are the multi-component representations of the Drinfel'd--Sokolov equations from our viewpoint,
from which we can clearly observe the coefficients in terms of the corresponding Cartan matrices on the left hand sides.
In order to explicitly write down the \ac{KdV}-type equations in the Drinfel'd--Sokolov classification,
we will have to make use of the constraints listed in the above subsection to decouple the multi-component systems.
Below we list these equations example by example.

\begin{example}
The continuum limits of the additional constraints in \eqref{A11:Constraint} yield the following equations:
\bse\label{A11:KdVMT}
\begin{align}
\partial_1(u_0+u_1)={}&-(u_0-u_1)^2, \label{A11:MTa} \\
\partial_3(u_0+u_1)={}&-2(u_0-u_1)\partial_1^{-1}\partial_3(u_0-u_1)+2(u_0-u_1)^4+3(u_0-u_1)^2\partial_1(u_0+u_1) \nonumber \\
&-3(\partial_1u_0)^2-3(\partial_1u_1)^2+9(\partial_1u_0)(\partial_1u_1)-4(u_0-u_1)\partial_1^2(u_0-u_1)-2\partial_1^3(u_0+u_1). \label{A11:MTb}
\end{align}
\ese
This set of equations plays a role of the non-auto \ac{BT} between potentials $u_0$ and $u_1$,
which can be used to decouple the two-component system \eqref{A11:KdV}.
To be more precise, \eqref{A11:KdVMT} can help to eliminate $u_0$ (resp. $u_1$) in equation \eqref{A11:KdVa} (resp. \eqref{A11:KdVb}).
Consequently, we find that both $u_0$ and $u_1$ satisfy the same closed-form scalar equation as follows:
\begin{align}\label{KdV}
\partial_3u=\frac{1}{4}\partial_1^3u+\frac{3}{2}(\partial_1u)^2,
\end{align}
which is nothing but the (potential) \ac{KdV} equation.
\end{example}

\begin{example}
We can similarly take the continuum limit of \eqref{A22:Constraint} and decouple \eqref{A22:KdV}.
As a result, we obtain the two separate scalar \ac{PDE}s of $u_0$ and $u_1$, given by
\begin{align}\label{SK}
\partial_5 u_0=-\frac{1}{9}\partial_1^5 u_0-\frac{5}{3}\left(\partial_1 u_0\right)\left(\partial_1^3 u_0\right)-\frac{5}{3}\left(\partial_1 u_0\right)^3
\end{align}
and
\begin{align}\label{KK}
\partial_5 u_1=-\frac{1}{9}\partial_1^5 u_1-\frac{5}{3}\left(\partial_1 u_1\right)\left(\partial_1^3 u_1\right)
-\frac{5}{3}\left(\partial_1 u_1\right)^3-\frac{5}{4}\left(\partial_1^2 u_1\right)^2,
\end{align}
which are the respective potential forms of the famous \ac{SK} and \ac{KK} equations.
\end{example}

\begin{example}
The continuum limits of the constraints in \eqref{C21:Constraint} help to eliminate $u_2$ in \eqref{C21:KdV}
and reduces the three-component system to a two-component system for $(u_0,u_1)$ given by
\bse\label{HS}
\begin{align}
\partial_3u_0={}&-\frac{1}{2}\partial_1^3u_0-\frac{3}{4}\partial_1^3u_1-\frac{3}{4}(\partial_1u_0)^2+\frac{3}{2}(\partial_1u_0)(\partial_1u_1)-\frac{9}{4}(\partial_1u_1)^2 \nonumber \\
&+3(\partial_1^2u_1)(u_0-u_1)+\frac{3}{2}(\partial_1u_0-3\partial_1u_1)(u_0-u_1)^2-\frac{3}{4}(u_0-u_1)^4,\\
\partial_3u_1={}&\frac{1}{4}\partial_1^3u_1-\frac{3}{4}(\partial_1u_0)^2-\frac{3}{2}(\partial_1u_0)(\partial_1u_1)+\frac{3}{4}(\partial_1u_1)^2
-\frac{3}{2}(\partial_1u_0+\partial_1u_1)(u_0-u_1)^2-\frac{3}{4}(u_0-u_1)^4.
\end{align}
\ese
We can alternatively eliminate $u_0$ and write down a coupled system for $(u_2,u_1)$, which takes exactly the same form as \eqref{HS}.
Equation \eqref{HS} is a deformed form of the Hirota--Satsuma equation.
\end{example}

\begin{example}
The continuous analogue of \eqref{D32:Constraint} can further be used to reduce the number of components in \eqref{D32:KdV}.
This leads to a two-component system for $(u_0,u_1)$ given by
\bse\label{Ito}
\begin{align}
\partial_3u_0={}&-2\partial_1^3u_0-\partial_1^3u_1-2(\partial_1u_0)^2+(\partial_1u_0)(\partial_1u_1)-\frac{5}{2}(\partial_1u_1)^2 \nonumber \\
&+3(\partial_1^2u_1)(u_0-u_1)+(\partial_1u_0-4\partial_1u_1)(u_0-u_1)^2-\frac{1}{2}(u_0-u_1)^4, \\
\partial_3u_1={}&-2(\partial_1u_0)^2-2(\partial_1u_0)(\partial_1u_1)+\frac{1}{2}(\partial_1u_1)^2-(2\partial_1u_0+\partial_1u_1)(u_0-u_1)^2-\frac{1}{2}(u_0-u_1)^4,
\end{align}
\ese
and also another two-component system for $(u_2,u_1)$ which takes exactly the same form of \eqref{Ito}.
Equation \eqref{Ito} is a deformed form of Ito's coupled \ac{KdV} equation.
\end{example}

The limits to the \ac{2D} Toda-type equations are the same as \eqref{Lim2}.
We let $n_1\rightarrow\infty$ and $p_1\rightarrow 0$ and introduce the change of variables from $(x_1,n_1)$ to $(x_1,x_{-1})$ composed of
\begin{align}\label{2DTLLim}
x_1\doteq x_1 \quad \hbox{and} \quad 2n_1p_1\doteq x_{-1}.
\end{align}
Then by series expansion, the respective leading terms in the continuum limits of \eqref{A11:u}, \eqref{A22:u}, \eqref{C21:u} and \eqref{D32:u}
give rise to the multi-component systems as follows:
\begin{description}
\item[$A_1^{(1)}$]
\begin{align}\label{A11:2DTL}
&\partial_1
\left(
\begin{array}{c}
\ln(1-\partial_{-1}u_0) \\
\ln(1-\partial_{-1}u_1)
\end{array}
\right)
=-
\begin{pmatrix}
2 & -2  \\
-2 & 2
\end{pmatrix}
\left(
\begin{array}{c}
u_0 \\
u_1
\end{array}
\right),
\end{align}
\item[$A_2^{(2)}$]
\begin{align}\label{A22:2DTL}
\partial_1
\left(
\begin{array}{c}
\ln(1-\partial_{-1}u_0) \\
\ln(1-\partial_{-1}u_1)
\end{array}
\right)
=-
\begin{pmatrix}
1 & -1 \\
-2 & 2
\end{pmatrix}
\left(
\begin{array}{c}
u_0 \\
u_1
\end{array}
\right),
\end{align}
\item[$C_2^{(1)}$]
\begin{align}\label{C21:2DTL}
\partial_1
\left(
\begin{array}{c}
\ln(1-\partial_{-1}u_0) \\
\ln(1-\partial_{-1}u_1) \\
\ln(1-\partial_{-1}u_2)
\end{array}
\right)
=-
\begin{pmatrix}
2 & -2 & 0 \\
-1 & 2 & -1   \\
0 & -2 & 2
\end{pmatrix}
\left(
\begin{array}{c}
u_0 \\
u_1 \\
u_2
\end{array}
\right),
\end{align}
\item[$D_3^{(2)}$]
\begin{align}\label{D32:2DTL}
\partial_1
\left(
\begin{array}{c}
\ln(1-\partial_{-1}u_0) \\
\ln(1-\partial_{-1}u_1)\\
\ln(1-\partial_{-1}u_2)
\end{array}
\right)
=-
\begin{pmatrix}
1 & -1 & 0 \\
-1 & 2 & -1 \\
0 & -1 & 1
\end{pmatrix}
\left(
\begin{array}{c}
u_0 \\
u_1 \\
u_2
\end{array}
\right).
\end{align}
\end{description}
These equations are equivalent forms of the \ac{2D} Toda-type equations.
Equations \eqref{A11:2DTL}, \eqref{A22:2DTL}, \eqref{C21:2DTL} and \eqref{D32:2DTL}, as expected, can alternatively be obtained
by performing \eqref{Ar1:Reduc}, \eqref{A2r2:Reduc}, \eqref{Cr1:Reduc} and \eqref{Dr2:Reduc} on \eqref{A:uLim2}, respectively.

We can also reduce the number of variables in the \ac{2D} Toda-type equations, by following the same procedure of deriving the \ac{KdV}-type equations,
from which the respective negative flows of the Drinfel'd--Sokolov hierarchies are obtained.
\begin{example}
In the $A_1^{(1)}$ class, by performing the limit \eqref{2DTLLim} on the second equation in \eqref{A11:Constraint} we obtain
\begin{align}\label{A11:2DTLMT}
(1-\partial_{-1}u_0)(1-\partial_{-1}u_1)=1.
\end{align}
Equation \eqref{A11:2DTLMT} allows us eliminate either $u_0$ or $u_1$ in \eqref{A11:2DTL}.
At the end, we find that in this class both $u_0$ and $u_1$ satisfy
\begin{align}\label{NegKdV}
\left[1-\partial_{-1}u-\frac{1}{2}\partial_1\partial_{-1}\ln(1-\partial_{-1}u)\right](1-\partial_{-1}u)=1,
\end{align}
i.e. the negative flow of the \ac{KdV} equation \eqref{KdV}.
\end{example}

\begin{example}
In the $A_2^{(2)}$ class, we can make use of the limit \eqref{A22:Constraint} to reduce the number of components in \eqref{A22:2DTL},
which leads to two separate equations given by
\begin{align}\label{NegSK}
\left[1-\partial_{-1}u_0-\partial_1\partial_{-1}\ln(1-\partial_{-1}u_0)\right](1-\partial_{-1}u_0)^2=1
\end{align}
and
\begin{align}\label{NegKK}
\left[1-\partial_{-1}u_1+\frac{1}{2}\partial_1\partial_{-1}\ln(1-\partial_{-1}u_1)\right]^2(1-\partial_{-1}u_1)=1,
\end{align}
respectively, which form the respective negative flows of the \ac{SK} equation \eqref{SK} and the \ac{KK} equation \eqref{KK}.
\end{example}

\begin{example}
In the $C_2^{(1)}$ class, the number of components in \eqref{C21:2DTL} can be reduced with the help of the continuum limit of \eqref{C21:Constraint}.
We then derive, for example, a two-component system for $(u_0,u_1)$ by eliminating $u_2$ as follows:
\bse\label{NegHS}
\begin{align}
&\partial_1\ln(1-\partial_{-1}u_0)=-2u_0+2u_1, \\
&\left[1+\partial_{-1}u_0-2\partial_{-1}u_1-\partial_1\partial_{-1}\ln(1-\partial_{-1}u_1)\right](1-\partial_{-1}u_0)(1-\partial_{-1}u_1)^2=1,
\end{align}
\ese
namely the negative flow of \eqref{HS}.
Equation \eqref{NegHS} can be decoupled by eliminating $u_1$, leading to a higher-order scalar equation for the potential $u_0$.
Alternatively, we can derive a coupled system for $(u_2,u_1)$, which is exactly the same as \eqref{NegHS}.
\end{example}

\begin{example}
In the $D_3^{(2)}$ class, we can also make use of the continuous analogue of \eqref{D32:Constraint} to reduce the number of components in \eqref{D32:2DTL}.
Consequently, we derive for $(u_0,u_1)$ a two-component system
\bse\label{NegIto}
\begin{align}
&\partial_1\ln(1-\partial_{-1}u_0)=-u_0+u_1, \\
&\left[1+\partial_{-1}u_0-2\partial_{-1}u_1+\partial_1\partial_{-1}\ln(1-\partial_{-1}u_1)\right](1-\partial_{-1}u_0)(1-\partial_{-1}u_1)=1,
\end{align}
\ese
i.e. the negative flow of \eqref{Ito}.
Since $u_0$ and $u_2$ are on the same footing in \eqref{D32:2DTL}, equation \eqref{NegIto} also holds for $(u_2,u_1)$.
Equation \eqref{NegIto} can be decoupled by eliminating $u_1$, which results in a higher-order scalar \ac{PDE} for $u_0$.
\end{example}

We have successfully discussed the two different continuum limits of the semi-discrete Drinfel'd--Sokolov equations.
Similarly, by performing the two limit schemes on the linear problems of the form \eqref{Laxb},
we are able to recover the linear problem \eqref{Laxa} and its higher-order counterparts in the series expansion for the \ac{KdV}-type equations
and the linear problem in terms of the flow variable $x_1$ for the \ac{2D} Toda-type equations.
Since these are known results, cf. e.g. \cite{FG1,FG83}, we omit the relevant formulae here.

\section{Concluding remarks}\label{S:Concl}

We constructed a large class of novel integrable semi-discrete equations within the \ac{DL} framework, in the language of infinite matrix.
In our scheme, the fundamental model is a new semi-discrete equation with two discrete and one continuous independent variables, i.e. \eqref{A:u},
which is integrable in the sense of possessing Lax representation.
This equation plays a role of discretisation of both the third-order differential-difference \ac{KP} equation and the \ac{2D} Toda equation.
Compared with the known results in the literature \cite{NCW85,DJM2},
an interesting observation is that the nonlocality occurs in the semi-discrete equation and its Lax pair.
In addition, equation \eqref{A:u} cannot be transferred into Hirota's bilinear form. Instead, its $\tau_n$-form, i.e. \eqref{A:tau}, is a homogeneous equation of degree $6$.

From our framework, the semi-discrete equation \eqref{A:u} is associated with the infinite-dimensional algebra $A_\infty$.
Thus, by performing various reductions on this master equation,
we successfully constructed the semi-discrete Drinfel'd--Sokolov hierarchies associated with the Kac--Moody algebras $A_r^{(1)}$, $A_{2r}^{(2)}$, $C_r^{(1)}$ and $D_{r+1}^{(2)}$.
Our result partly solved the unsolved problem proposed by Date, Jimbo and Miwa in \cite{DJM5} in their series work,
though we firmly believe that it could be solved in their framework.
We have also shown that all these semi-discrete Drinfel'd--Sokolov equations have a unified Lax structure \eqref{Lax},
in which the $n_1$-part provides a new discrete version of the theory of factorisation of operators proposed by Fordy and Gibbons \cite{FG80,FG83}.

The semi-discrete modified Drinfel'd--Sokolov hierarchies seem not to exist in our framework.
This is because in the higher-dimensional case, the nonlocality in terms of the discrete variable $n$ cannot be avoided in the semi-discrete modified equation \eqref{A:v},
resulting in the consequence that an divergence issue appears when we perform periodic reductions.
This also coincides the statement made by Adler and Postnikov \cite{AP11} that there may not exist a Miura transform between \ac{SK} and \ac{KK},
when they studied the semi-discrete \ac{SK} equation (which differs from ours in this paper) proposed in \cite{TH96};
while such an issue does not occur in the continuous case,
guaranteed by the existence of the \ac{FG} equation, i.e. an equation as the modification of both \ac{SK} and {KK}, see \cite{FG1}.

Our construction also naturally induces exact solutions for the obtained semi-discrete equations.
The procedure follows from the original idea from Fokas and Ablowitz \cite{FA81}.
In fact, by specifying integration measure and integration domain, we are able to construct special classes of explicit solutions.
For instance, the Cauchy matrix type solution (i.e. finite-pole solution) can be easily obtained by taking a special integration measure containing a finite number of poles.
Since the reductions that we perform in this paper coincide with those in the theory of the \ac{2D} Toda-type equations,
we refer the reader to \cite{Fu21b} for the general formulae of the Cauchy matrix solutions of all the discussed semi-discrete Drinfel'd--Sokolov equations
by substituting the plane wave factors $\rho_n(k)$ and $\sigma_n(k')$ with \eqref{A:PWF}.

Finally, we would like to point out that the \ac{DL} framework allows us to search for the fully discrete Drinfel'd--Sokolov hierarchies.
However, a number of nontrivial techniques are involved in the construction of closed-form integrable discrete equations,
which also leads to very different integrability characteristics.
For this reason, we shall present those results separately.

\section*{Acknowledgments}
This project was supported by the National Natural Science Foundation of China (grant no. 11901198) and Shanghai Pujiang Program (grant no. 19PJ1403200).
WF was also partially sponsored by the Science and Technology Commission of Shanghai Municipality (grant no. 18dz2271000).

\begin{appendix}

\section{Derivation of the infinite matrix relations and closed-form equations}

\subsection{Derivation of \eqref{A:UDyn}}\label{A:UDynPf}
We first derive the dynamical evolution of $\bU_n$ in terms of the continuous independent variable $x_1$.
The derivative of \eqref{U} with respect to $x_1$ provides us with
\begin{align*}
\partial_1\bU_n=(1-\bU_n\bOa)(\partial_1\bC_n)-(\partial_1\bU_n)\bOa\bC_n,
\end{align*}
which can alternatively be written as
\begin{align*}
(\partial_1\bU_n)(1+\bOa\bC_n)=\bLd\bC_n+\bU_n\tbLd-\bU_n\bOa\bLd\bC_n.
\end{align*}
If we substitute $\bOa\bLd$ with $\bO-\tbLd\bOa$ by following \eqref{A:OaDyn}, the above equation turns out to be
\begin{align*}
(\partial_1\bU_n)(1+\bOa\bC_n)=\bLd\bC_n+\bU_n\tbLd(1+\bOa\bC_n)-\bU_n\bO\bC_n.
\end{align*}
In other words, $\bU_n$ obeys the dynamical evolution \eqref{A:UDyna}.
We next derive the discrete evolution of $\bU_n$ with respect to $n_1$.
Performing the $\wt\cdot$ operation on \eqref{U} and simultaneously multiplying it by $\frac{p_1-\tbLd}{p_1+\tbLd}$, we obtain
\begin{align*}
\wt\bU_n\frac{p_1-\tbLd}{p_1+\tbLd}=(1-\wt\bU_n\bOa)\wt\bC_n\frac{p_1-\tbLd}{p_1+\tbLd}=\frac{p_1+\bLd}{p_1-\bLd}\bC_n-\wt\bU_n\bOa\frac{p_1+\bLd}{p_1-\bLd}\bC_n,
\end{align*}
where in the last step we have made use of \eqref{A:CDynb}. Notice that \eqref{A:OaDyn} implies that
\begin{align*}
\bOa\frac{p_1+\bLd}{p_1-\bLd}-\frac{p_1-\tbLd}{p_1+\tbLd}\bOa=2p_1\frac{1}{p_1+\tbLd}\bO\frac{1}{p_1-\bLd}.
\end{align*}
We derive that the infinite matrix $\bU_n$ satisfies
\begin{align*}
\wt\bU_n\frac{p_1-\tbLd}{p_1+\tbLd}=\frac{p_1+\bLd}{p_1-\bLd}\bC_n-\wt\bU_n\left[\frac{p_1-\tbLd}{p_1+\tbLd}\bOa+2p_1\frac{1}{p_1+\tbLd}\bO\frac{1}{p_1-\bLd}\right]\bC_n,
\end{align*}
which is equivalent to \eqref{A:UDynb}.
The derivation of \eqref{A:UDync} is very similar to that of \eqref{A:UDynb}, by shifting \eqref{U} by one unit with respect to $n$.

\subsection{Derivation of \eqref{A:tauDyn}}\label{A:tauDynPf}

First of all, we differentiate $\ln\tau_n$ with respect to $x_1$ and obtain
\begin{align*}
\partial_1\ln\tau_n=\partial_1\ln[\det(1+\bOa\bC_n)]=\partial_1\tr[\ln(1+\bOa\bC_n)]=\tr[\partial_1\ln(1+\bOa\bC_n)]=\tr[(1+\bOa\bC_n)^{-1}\bOa(\partial_1\bC_n)],
\end{align*}
in which the identity $\ln\det[\cdot]=\tr\ln[\cdot]$ is used for the second equality.
Notice that $C_n$ satisfies \eqref{A:CDyna}, this equations can be further rewritten as
\begin{align*}
\partial_1\ln\tau_n=\tr[(1+\bOa\bC_n)^{-1}\bOa(\bLd\bC_n+\bC_n\tbLd)]=\tr[(1+\bOa\bC_n)^{-1}\bOa\bLd\bC_n+(1+\bOa\bC_n)^{-1}\bOa\bC_n\tbLd].
\end{align*}
By replacing $\bOa\bLd$ with the help of \eqref{A:OaDyn}, we end up with
\begin{align*}
\partial_1\ln\tau_n=\tr[\bC_n(1+\bOa\bC_n)^{-1}\bOa\bLd+\bC_n(1+\bOa\bC_n)^{-1}\tbLd\bOa]=\tr(\bU_n\bO)=\tr(\bO\bU_n)=\bU_n^{(0,0)};
\end{align*}
in other words, equation \eqref{A:tauDyna} is proven.
Secondly, performing the tilde shift operation on the tau function and simultaneously taking \eqref{A:CDynb} and \eqref{A:OaDyn} into consideration, we have
\begin{align*}
\wt\tau_n=\det(1+\bOa\wt\bC_n)=\det\left(1+\bOa\frac{p_1+\bLd}{p_1-\bLd}\bC_n\frac{p_1+\tbLd}{p_1-\tbLd}\right)
=\det\left(1+\bOa\bC_n+2p_1\frac{1}{p_1-\tbLd}\bO\frac{1}{p_1-\bLd}\bC_n\right),
\end{align*}
which is equivalent to
\begin{align*}
\wt\tau_n=\det(1+\bOa\bC_n)\det\left(1+2p_1(1+\bOa\bC_n)^{-1}\frac{1}{p_1-\tbLd}\bO\frac{1}{p_1-\bLd}\bC_n\right)
=\tau_n\left[1+2p_1\left(\frac{1}{p_1-\bLd}\bU_n\frac{1}{p_1-\tbLd}\right)^{(0,0)}\right],
\end{align*}
because of the Weinstein--Aronszajn formula as well as the relation \eqref{U}.
Hence, we derive the first equation in \eqref{A:tauDynb}.
The second equality in \eqref{A:tauDynb} is proven by performing $\ut\cdot$ operation on the tau function through the same procedure.
Following the idea of deriving \eqref{A:tauDynb}, we can similarly preform the shift operation with regard to $n$ on the tau function.
This gives rise to the two equalities in \eqref{A:tauDync}.

\subsection{Derivation of \eqref{A:tau}}\label{A:tauPf}

The first few relations that we need are the ones derived from \eqref{A:UDyna}.
To find equations involving $V_n(a)$, $W_n(b)$ and $S_n(a,b)$, we consider the operations $[\eqref{A:UDyna}\frac{1}{a+\tbLd}]^{(0,0)}$, $[\frac{1}{a+\bLd}\eqref{A:UDyna}]^{(0,0)}$ and $[\frac{1}{a+\bLd}\eqref{A:UDyna}\frac{1}{b+\tbLd}]^{(0,0)}$, one by one.
These yield the respective dynamical evolutions of $V_n(a)$, $W_n(b)$ and $S_n(a,b)$ with regard to $x_1$ as follows:
\bse\label{x1}
\begin{align}
&\partial_1V_n(a)=a-aV_n(a)-\left(\bLd\bU_n\frac{1}{a+\tbLd}\right)^{(0,0)}-u_nV_n(a), \label{x1:V}\\
&\partial_1W_n(a)=a-aW_n(a)-\left(\frac{1}{a+\bLd}\bU_n\tbLd\right)^{(0,0)}-u_nW_n(a), \label{x1:W}\\
&\partial_1S_n(a,b)=1-W_n(a)V_n(b)-(a+b)S_n(a,b). \label{x1:S}
\end{align}
\ese
Next, we derive the necessary equations involving $u_n$, $V_n(a)$ and $W_n(a)$ which describe the evolution with respect to the discrete variable $n_1$.
By evaluating $\eqref{A:UDynb}^{(0,0)}$, $[\eqref{A:UDynb}\frac{1}{a+\tbLd}]^{(0,0)}$ as well as $[\frac{1}{a+\bLd}\eqref{A:UDynb}]^{(0,0)}$, the following equations arise:
\bse\label{n1}
\begin{align}
&2p_1+u_n-\wt u_n=2p_1\wt V_n(p_1)W_n(-p_1), \label{n1:u} \\
&\frac{V_n(a)}{\wt V_n(p_1)}+\frac{p_1+a}{p_1-a}\frac{\wt V_n(a)}{\wt V_n(p_1)}=\frac{2p_1}{p_1-a}+2p_1S_n(-p_1,a), \label{n1:V} \\
&\frac{\wt W_n(a)}{W_n(-p_1)}+\frac{p_1-a}{p_1+a}\frac{W_n(a)}{W_n(-p_1)}=\frac{2p_1}{p_1+a}-2p_1\wt S_n(a,p_1). \label{n1:W}
\end{align}
\ese
Finally, we also need the relations involving the evolutions in terms of the discrete direction $n$.
This can be realised by taking $[\eqref{A:UDync}\frac{1}{a+\tbLd}]^{(0,0)}$, $[\frac{1}{a+\bLd}\eqref{A:UDync}]^{(0,0)}$ and also $[\frac{1}{a+\bLd}\eqref{A:UDync}\frac{1}{b+\tbLd}]^{(0,0)}$, respectively.
As a result, we achieve equations as follows:
\bse\label{n}
\begin{align}
&a-aV_{n+1}(a)=\left(\bLd\bU_n\frac{1}{a+\tbLd}\right)^{(0,0)}+u_{n+1}V_n(a), \label{n:V}\\
&a-aW_n(a)=\left(\frac{1}{a+\bLd}\bU_{n+1}\tbLd\right)^{(0,0)}+u_nW_{n+1}(a), \label{n:W}\\
&W_{n+1}(a)V_n(b)=1-aS_n(a,b)-bS_{n+1}(a,b). \label{n:S}
\end{align}
\ese
The above nine equations, i.e. \eqref{x1}, \eqref{n1} and \eqref{n} form fundamental ingredients to construct a scalar equation for the tau function.

Now we establish their links with the tau function.
If we take $a=b=-p_1$ in \eqref{x1:S} and \eqref{n:S}, respectively, these two equations are reformulated as
\begin{align*}
V_n(-p_1)W_n(-p_1)=1+2p_1S_n(-p_1,-p_1)-\partial_1S_n(-p_1,-p_1)=1+2p_1S_n(-p_1,-p_1)-\frac{1}{2p_1}\partial_1\left[1+2p_1S_n(-p_1,-p_1)\right]
\end{align*}
and
\begin{align*}
V_n(-p_1)W_{n+1}(-p_1)=1+p_1S_n(-p_1,-p_1)+p_1S_{n+1}(-p_1,-p_1)=\frac{1}{2}\left[1+2p_1S_n(-p_1,-p_1)+1+2p_1S_{n+1}(-p_1,-p_1)\right],
\end{align*}
respectively.
Recall that \eqref{A:tauDynb} provides a direct transformation between $\tau_n$ and $S_n(a,b)$.
From the above two equations, we obtain
\begin{align}\label{A:tauVWa}
V_n(-p_1)W_n(-p_1)=\frac{1}{2p_1}\frac{\wt{\tau}_n}{\tau_n}\left(2p_1+\partial_1\ln\frac{\tau_n}{\wt\tau_n}\right) \quad \hbox{and} \quad V_n(-p_1)W_{n+1}(-p_1)=\frac{1}{2}\left(\frac{\wt\tau_n}{\tau_n}+\frac{\wt\tau_{n+1}}{\tau_{n+1}}\right),
\end{align}
which further induce the transformations
\begin{align}\label{A:tauVWb}
\frac{V_{n+1}(-p_1)}{V_n(-p_1)}=\frac{\displaystyle 2p_1+\partial_1\ln\frac{\tau_{n+1}}{\wt\tau_{n+1}}}{\displaystyle p_1\left(1+\frac{\wt{\tau}_n\tau_{n+1}}{\tau_n\wt{\tau}_{n+1}}\right)} \quad \hbox{and} \quad
\frac{W_{n-1}(-p_1)}{W_n(-p_1)}=\frac{\displaystyle 2p_1+\partial_1\ln\frac{\tau_{n-1}}{\wt\tau_{n-1}}}{\displaystyle p_1\left(1+\frac{\wt{\tau}_n\tau_{n-1}}{\tau_n\wt{\tau}_{n-1}}\right)}.
\end{align}
Adding \eqref{x1:V} and \eqref{n:V} up, we are able to eliminate $\left(\bLd\bU_n\frac{1}{a+\tbLd}\right)^{(0,0)}$ and derive
\begin{align*}
\partial_1\ln V_n(a)=u_{n+1}-u_n-a\left(1-\frac{V_{n+1}(a)}{V_n(a)}\right), \quad \hbox{and similarly we have} \quad \partial_1\ln W_n(a)=u_{n-1}-u_n-a\left(1-\frac{W_{n-1}(a)}{W_n(a)}\right),
\end{align*}
which follows from \eqref{x1:W} and \eqref{n:W}.
Setting $a=-p_1$ in these two equations, we immediately obtain
\begin{align}\label{x1-n:VW}
\partial_1\ln V_n(-p_1)=u_{n+1}-u_n+p_1\left(1-\frac{V_{n+1}(-p_1)}{V_n(-p_1)}\right) \quad \hbox{and} \quad \partial_1\ln W_n(-p_1)=u_{n-1}-u_n+p_1\left(1-\frac{W_{n-1}(-p_1)}{W_n(-p_1)}\right).
\end{align}
We compute the addition of the two equations in \eqref{x1-n:VW}, which leads to an identity
\begin{align*}
\partial_1\ln[V_n(-p_1)W_n(-p_1)]=p_1\left(2-\frac{V_{n+1}(-p_1)}{V_n(-p_1)}-\frac{W_{n-1}(-p_1)}{W_n(-p_1)}\right)+u_{n+1}-2u_n+u_{n-1}.
\end{align*}
Notice that the transformations \eqref{A:tauDyna}, \eqref{A:tauVWa} and \eqref{A:tauVWb} allow us to substitute $u_n$, $V_n(-p_1)W_n(-p_1)$, $\frac{V_{n+1}(-p_1)}{V_n(-p_1)}$ and $\frac{W_{n-1}(-p_1)}{W_n(-p_1)}$ with the tau function.
We finally reach to the scalar closed-form equation of \eqref{A:tau}.

Equation \eqref{A:tau} is also obtainable by setting $a=b=p_1$ in the above derivation.
In fact, when $a=b=p_1$ we have the transformations
\begin{align}\label{A:tauVWa'}
W_n(p_1)V_n(p_1)=\frac{1}{2p_1}\frac{\ut{\tau}_n}{\tau_n}\left(2p_1+\partial_1\ln\frac{\ut{\tau}_n}{\tau_n}\right) \quad \hbox{and} \quad W_{n+1}(p_1)V_n(p_1)=\frac{1}{2}\left(\frac{\ut\tau_n}{\tau_n}+\frac{\ut\tau_{n+1}}{\tau_{n+1}}\right),
\end{align}
as well as
\begin{align}\label{A:tauVWb'}
\frac{V_{n+1}(p_1)}{V_n(p_1)}=\frac{\displaystyle 2p_1+\partial_1\ln\frac{\ut\tau_{n+1}}{\tau_{n+1}}}{\displaystyle p_1\left(1+\frac{\ut\tau_n\tau_{n+1}}{\tau_n\ut\tau_{n+1}}\right)} \quad \hbox{and} \quad
\frac{W_{n-1}(p_1)}{W_n(p_1)}=\frac{\displaystyle 2p_1+\partial_1\ln\frac{\ut\tau_{n-1}}{\tau_{n-1}}}{\displaystyle p_1\left(1+\frac{\ut\tau_n\tau_{n-1}}{\tau_n\ut\tau_{n-1}}\right)}.
\end{align}
Simultaneously, we have also
\begin{align}\label{x1-n:VW'}
\partial_1\ln V_n(p_1)=u_{n+1}-u_n-p_1\left(1-\frac{V_{n+1}(p_1)}{V_n(p_1)}\right) \quad \hbox{and} \quad \partial_1\ln W_n(p_1)=u_{n-1}-u_n-p_1\left(1-\frac{W_{n-1}(p_1)}{W_n(p_1)}\right).
\end{align}
Then the identity
\begin{align*}
\partial_1\ln[V_n(p_1)W_n(p_1)]=u_{n+1}-2u_n+u_{n-1}+p_1\left(\frac{V_{n+1}(p_1)}{V_n(p_1)}+\frac{W_{n-1}(p_1)}{W_n(p_1)}-2\right)
\end{align*}
yields equation \eqref{A:tau}.

\subsection{Derivation of \eqref{A:uDyn}}\label{A:uDynPf}
Recall that $\bU_n$ satisfies \eqref{A:UDyna}.
By differentiating \eqref{u} with respect to $x_1$, we have
\begin{align*}
\partial_1\bu_n(k)={}&-(\partial_1\bU_n)\bOa\bc(k)\rho_n(k)+(1-\bU_n\bOa)\bc(k)[\partial_1\rho_n(k)] \\
={}&-(\bLd\bU_n+\bU_n\tbLd-\bU_n\bO\bU_n)\bOa\bc(k)\rho_n(k)+(1-\bU_n\bOa)\bLd\bc(k)\rho_n(k) \\
={}&\bLd(1-\bU_n\bOa)\bc(k)\rho_n(k)-\bU_n(\bOa\bLd+\tbLd\bOa)\bc(k)\rho_n(k)+\bU_n\bO\bU_n\bOa\bc(k)\rho_n(k).
\end{align*}
Notice that $\bOa$ satisfies \eqref{A:OaDyn}, we thus obtain
\begin{align*}
\partial_1\bu_n(k)=\bLd(1-\bU_n\bOa)\bc(k)\rho_n(k)-\bU_n\bO(1-\bU_n\bOa)\bc(k)\rho_n(k).
\end{align*}
Making use of \eqref{u} again, we end up with \eqref{A:uDyna}.
Next, we derive the dynamical evolution of $\bu_n(k)$ in terms of $n_1$.
Performing the tilde shift on \eqref{A:u} provides us with
\begin{align*}
\wt\bu_n(k)={}&(1-\wt\bU_n\bOa)\bc(k)\wt\rho_n(k)=(1-\wt\bU_n\bOa)\frac{p_1+\bLd}{p_1-\bLd}\bc(k)\rho_n(k)
=\frac{p_1+\bLd}{p_1-\bLd}\bc(k)\rho_n(k)-\wt\bU_n\bOa\frac{p_1+\bLd}{p_1-\bLd}\bc(k)\rho_n(k) \\
={}&\frac{p_1+\bLd}{p_1-\bLd}\bc(k)\rho_n(k)-\wt\bU_n\left[\frac{p_1-\tbLd}{p_1+\tbLd}\bOa+2p_1\frac{1}{p_1+\tbLd}\bO\frac{1}{p_1-\bLd}\right]\bc(k)\rho_n(k).
\end{align*}
Notice that $\bU_n$ satisfies \eqref{A:UDynb}. We are able to reformulate this equation as
\begin{align*}
\wt\bu_n(k)=\frac{p_1+\bLd}{p_1-\bLd}(1-\bU_n\bOa)\bc(k)\rho_n(k)-2p_1\wt\bU_n\frac{1}{p_1+\tbLd}\bO\frac{1}{p_1-\bLd}(1-\bU_n\bOa)\bc(k)\rho_n(k),
\end{align*}
which is nothing but \eqref{A:uDynb}.
Following a similar procedure of deriving \eqref{A:uDynb}, we can derive \eqref{A:uDync} in virtue of \eqref{A:UDync}.

\subsection{Derivation of \eqref{A:Lax}}\label{A:LaxPf}
We subtract \eqref{A:uDync} from \eqref{A:uDyna} and obtain
\begin{align*}
\partial_1\bu_n(k)=\bu_{n+1}(k)+(\bU_{n+1}-\bU_n)\bO\bu_n(k),
\end{align*}
whose $0$th-component gives us the linear equation \eqref{A:Laxa}.
To derive equation \eqref{A:Laxb}, we evaluate $\eqref{A:uDynb}^{(0)}$ and $[\frac{1}{p_1-\bLd}\eqref{A:uDync}]^{(0)}$ and obtain
\begin{align*}
\wt\phi_n+\phi_n=2p_1\wt V_{n}(p_1)\left(\frac{1}{p_1-\bLd}\bu_n(k)\right)^{(0)} \quad \hbox{and} \quad
p_1\left(\frac{1}{p_1-\bLd}\bu_n(k)\right)^{(0)}-\left(\frac{1}{p_1-\bLd}\bu_{n+1}(k)\right)^{(0)}=W_{n+1}(-p_1)\phi_n.
\end{align*}
These two equations result in a linear equation in terms of $\phi_n$ taking the form of
\begin{align*}
\frac{\wt\phi_n+\phi_n}{2\wt V_n(p)W_{n+1}(-p_1)}-\frac{\wt\phi_{n+1}+\phi_{n+1}}{2p_1\wt V_{n+1}(p_1)W_{n+1}(-p_1)}=\phi_n.
\end{align*}
Notice that setting $a=-p_1$ in \eqref{n1:V} and $b=p_1$ \eqref{n1:W}, respectively yield
\begin{align*}
\frac{V_n(-p_1)}{\wt{V}_n(p_1)}=\frac{\wt{\tau}_n}{\tau_n} \quad \hbox{and} \quad \frac{W_n(-p_1)}{\wt{W}_n(p_1)}=\frac{\wt{\tau}_n}{\tau_n}.
\end{align*}
We finally reach to \eqref{A:Laxb}, in virtue of \eqref{A:tauVWa} and \eqref{n1:u}, where $\tau_n$ is replaced by $u_n$ through the transformation $u_n=\partial_1\ln\tau_n$.
\end{appendix}

\renewcommand{\bibname}{References}
\bibliography{References}
\bibliographystyle{plain}

\end{document}